\documentclass[11pt]{article}
\usepackage[
  a4paper,
  top=2.2cm, bottom=2.2cm,
  left=2.5cm, right=2.5cm
]{geometry}
\usepackage{times}
\usepackage{soul}
\usepackage{url}
\usepackage[hidelinks]{hyperref}
\usepackage[utf8]{inputenc}
\usepackage[small]{caption}
\usepackage{graphicx}
\usepackage{amsmath}
\usepackage{amsthm}
\usepackage{booktabs}
\usepackage[numbers]{natbib}
\usepackage{algorithm}
\usepackage[switch]{lineno}

\usepackage[noend]{algpseudocode}
\usepackage{enumitem}
\usepackage{amsfonts}
\usepackage{mathtools}
\usepackage{complexity}
\usepackage{color}
\usepackage{mleftright}
\usepackage{multicol,multirow}

\newtheorem{theorem}{Theorem}
\newtheorem{lemma}{Lemma}

\newtheorem{definition}{Definition}

\algnewcommand{\IfSingleLine}[2]{\State \algorithmicif\ #1\ \algorithmicthen #2}
\algnewcommand{\ElseIfSingleLine}[2]{\State \algorithmicelse\ \algorithmicif\ #1\ \algorithmicthen #2}
\algnewcommand{\algorithmicand}{\textbf{ and }}
\algnewcommand{\algorithmicor}{\textbf{ or }}
\algnewcommand{\OR}{\algorithmicor}
\algnewcommand{\AND}{\algorithmicand}

\newcommand{\ie}{\textit{i.e.}}
\newcommand{\bigO}{\mathcal{O}}
\newcommand{\bigOs}{\bigO^*}
\newcommand{\card}[1]{\left|#1\right|}
\newcommand{\nl}{\bar}
\newcommand{\formula}{F}
\newcommand{\literal}{\ell}
\newcommand{\model}{\sigma}
\newcommand{\wfunc}{w}

\newcommand{\solW}{W}
\newcommand{\ins}{\mathcal{I}}
\newcommand{\shortins}{(\formula, \wfunc)}
\newcommand{\longins}{(\formula, \wfunc, \solW)}
\newcommand{\reduced}[1]{R(#1)}

\newcommand{\algtwocnf}{\texttt{Alg2CNF}}
\newcommand{\algthreecnf}{\texttt{Alg3CNF}}
\newcommand{\algpripw}{\texttt{AlgPrimalPw}}
\newcommand{\algdualpw}{\texttt{AlgDualPw}}
\newcommand{\algbf}{\texttt{Brute}}

\DeclareMathOperator{\pw}{pw}
\DeclareMathOperator{\var}{var}
\DeclareMathOperator{\ltr}{ltr}

\DeclareMathOperator{\WMC}{WMC}

\newtheorem{rrule}{R-Rule}

\title{New Algorithms for \#2-SAT and \#3-SAT}

\author{
Junqiang Peng$^1$,
Zimo Sheng$^2$,
Mingyu Xiao$^{1}$\footnote{Corresponding Author},
\\
{\footnotesize
    \textsuperscript{\rm 1}
University of Electronic Science and Technology of China, Chengdu, China
}
		\\
{\footnotesize
    \textsuperscript{\rm 2} Anhui University, China
}
		\\
{\footnotesize
		jqpeng0@foxmail.com, shengzimo2016@gmail.com, myxiao@uestc.edu.cn
}
}
\date{}

\begin{document}

\maketitle

\begin{abstract}
    The \#2-SAT and \#3-SAT problems involve counting the number of satisfying assignments (also called models) for instances of 2-SAT and 3-SAT, respectively. In 2010, Zhou et al. proposed an $\bigOs(1.1892^m)$-time algorithm for \#2-SAT and an efficient approach for \#3-SAT, where $m$ denotes the number of clauses. In this paper, we show that the weighted versions of \#2-SAT and \#3-SAT can be solved in $\bigOs(1.1082^m)$ and $\bigOs(1.4423^m)$ time, respectively.
    These results directly apply to the unweighted cases and achieve substantial improvements over the previous results. These advancements are enabled by the introduction of novel reduction rules, a refined analysis of branching operations, and the application of path decompositions on the primal and dual graphs of the formula.
\end{abstract}

\section{Introduction}
The \textsc{Boolean Satisfiability} problem (SAT), the first problem to be proven $\NP$-complete~\citep{conf/stoc/Cook71}, is a cornerstone of computational complexity theory.
 Its counting variant, the \textsc{Propositional Model Counting} problem (\#SAT), introduced and shown to be $\#\P$-complete by \citet{journals/tcs/Valiant79}, holds comparable significance.

Given a CNF formula, SAT seeks to determine whether the formula is satisfiable, while \#SAT aims to count the number of satisfying assignments (also known as models). Both SAT and \#SAT, along with their variants, are among the most influential problems in computational theory due to their broad applications in computer science, artificial intelligence, and numerous other domains, both theoretical and practical. As such, these problems have garnered substantial attention and have been extensively studied across various fields, including computational complexity and algorithm design. For comprehensive surveys, For comprehensive surveys, we refer to~\citep{series/faia/HandbookSAT,journals/cacm/FichteBHS23}.

This paper focuses on \#SAT and its weighted extension, \textsc{Weighted Model Counting} (WMC, or weighted \#SAT). In WMC, each literal (a variable or its negation) in the formula is assigned a weight. The task is to compute the sum of the weights of all satisfying assignments, where the weight of an assignment is the product of the weights of the true literals in the assignment. Notably, WMC reduces to \#SAT when all literals have identical weights. Efficient algorithms for (weighted) \#SAT have a profound impact on various application areas~\citep{series/faia/GomesSS21}, such as probabilistic inference~\citep{journals/ai/Roth96,conf/focs/BacchusDP03,conf/aaai/SangBK05,journals/ai/ChaviraD08}, network reliability estimation~\citep{conf/aaai/Duenas-OsorioMP17}, and explainable AI~\citep{conf/sat/NarodytskaSMIM19}. This significance is highlighted by the annual Model Counting Competition\footnote{\url{https://mccompetition.org/}}, which bridges theoretical advancements and practical implementations in model counting.

A fundamental question is: how fast (weighted) \#SAT can be solved in the worst case?  
The naïve algorithm, which enumerates all assignments, runs in $\bigOs(2^n)$ time\footnote{The $\bigOs$ notation suppresses polynomial factors in the input size.}, where \( n \) is the number of variables in the formula. Under the Strong Exponential Time Hypothesis (SETH)~\citep{journals/jcss/ImpagliazzoP01}, SAT (and thus \#SAT) cannot be solved in $\bigOs((2-\epsilon)^n)$ time for any constant \(\epsilon > 0\). Another key parameter, the number of clauses \( m \) in the formula, has also been extensively studied. It is well-known that \#SAT can be solved in $\bigOs(2^m)$ time using the Inclusion-Exclusion principle~\citep{journals/siamcomp/Iwama89,journals/ipl/Lozinskii92}, which also applies to the weighted variant. However, no algorithm with a runtime of $\bigOs(c^m)$ for \( c < 2 \) was discovered. In \citeyear{journals/talg/CyganDLMNOPSW16}, \citet{journals/talg/CyganDLMNOPSW16} proved that such an algorithm does not exist unless SETH fails.

The barriers $2^n$ and $2^m$ can be broken for restricted versions of (weighted) \#SAT.
One notable example is the (weighted) \#$k$-SAT problem, where each clause contains at most \( k \) literals in the input formula. There is a rich history of developing faster algorithms for \#2-SAT and \#3-SAT~\citep{journals/tcs/Dubois91,journals/tcs/Zhang96a,journals/tcs/DahllofJW05,journals/ipl/Kutzkov07,conf/iwpec/Wahlstrom08,conf/aaai/ZhouYZ10}. The current fastest algorithms for weighted \#2-SAT and \#3-SAT achieve runtimes of $\bigOs(1.2377^n)$~\citep{conf/iwpec/Wahlstrom08} and $\bigOs(1.6423^n)$~\citep{journals/ipl/Kutzkov07}, respectively. For the parameter \( m \), \citeauthor{conf/aaai/ZhouYZ10} introduced an $\bigOs(1.1740^m)$-time algorithm for \#2-SAT and suggested a simple approach for \#3-SAT.
However, the analysis of the \#3-SAT algorithm in \citep{conf/aaai/ZhouYZ10} does not yield a valid runtime bound (more details can be found in  Appendix~\ref{appendix:note}).

\paragraph{Our Contribution.}
We propose two novel algorithms, $\algtwocnf$ and $\algthreecnf$, for weighted \#2-SAT and weighted \#3-SAT, achieving runtime bounds of $\bigOs(1.1082^m)$ and $\bigOs(1.4423^m)$, respectively. 
The algorithms and complexity bounds directly apply to the unweighted case, significantly improving upon previous results.
To this end, we bring new techniques for (weighted) \#2-SAT and \#3-SAT.

\paragraph{Our Approach.}
Most existing algorithms, including ours, are classical branch-and-search algorithms (also called DPLL-style algorithms) that first apply reduction (preprocessing) rules and then recursively solve the problem via branching (e.g., selecting a variable and assigning it values).
Typically, variables are processed in descending order of their degree, where the \emph{degree} of a variable is the number of its occurrences in the formula.
However, such algorithms often perform poorly when encountering low-degree variables.
To address this, previous work has employed tailored branching strategies with intricate analyses to mitigate this bottleneck.

Our approach departs from these complexities. Instead of elaborate branching, we apply path decompositions on the primal and dual graphs of the formula to efficiently handle low-degree cases. The primal and dual graphs represent structural relationships between variables and clauses, and path decompositions allow us to transform these structures into a path-like form. Although our algorithms rely on simple branching, we demonstrate through sophisticated analyses that this approach, combined with our reduction rules, achieves substantial efficiency. These ideas may hold potential for broader applications in the future.

\paragraph{Other Related Works.} 
There is extensive research on fast algorithms for SAT and its related problems parameterized by $m$.
We list the current best results for some of these problems, achieved after a series of improvements.
\citet{conf/aaai/ChuXZ21} showed that SAT can be solved in $\bigOs(1.2226^m)$ time. 
\citet{journals/jal/BeigelE05} introduced an $\bigOs(1.3645^m)$-time algorithm for 3-SAT.
We also mention the \textsc{Maximum Satisfiability} problem (MaxSAT), an optimization version of SAT, where the objective is to satisfy the maximum number of clauses in a given formula.
Currently, MaxSAT and Max-2-SAT can be solved in time $\bigOs(1.2886^m)$~\citep{conf/ijcai/Xiao22} and $\bigOs(1.1159^m)$~\citep{conf/soda/GaspersS09}, respectively.

\section{Preliminaries}
\subsection{Notations}
A \emph{Boolean variable} (or simply \emph{variable}) can be assigned value $1$ (\textsc{true}) or $0$ (\textsc{false}).
A variable $x$ has two corresponding \emph{literals}: the positive literal $x$ and the negative literal $\nl{x}$.
We use $\nl{x}$ to denote the negation of literal $x$, and thus $\nl{\nl{x}}=x$.
Let $V$ be a set of variables.
A \emph{clause} on $V$ is a set of literals on $V$. Note that a clause might be empty.
A \emph{CNF formula} (or simply \emph{formula}) over $V$ is a set of clauses on $V$.
We denote by $\var(\formula)$ the variable set of $\formula$.
For a literal $\literal$, $\var(\literal)$ denotes its corresponding variable.
For a clause $C$, $\var(C)$ denotes the set of variables such that either $x\in C$ or $\nl{x}\in C$.
We denote by $n(\formula)$ and $m(\formula)$ the number of variables and clauses in formula $\formula$, respectively.

An \emph{assignment} for variable set $V$ is a mapping $\model: V\rightarrow\{0, 1\}$.
Given an assignment $\model$, a clause is \emph{satisfied} by $\model$ if at least one literal in it gets value $1$ under $\model$.
An assignment for $\var(\formula)$ is called a \emph{model} of $\formula$ if $\model$ satisfies all clauses in $\formula$.
We write $\model \models \formula$ to indicate that $\model$ is a model of $\formula$. 

\begin{definition}[Weighted Model Count]
    Let $\formula$ be a formula, $\ltr(\formula):=\bigcup_{x\in \var(\formula)}\{x, \nl{x}\}$ be the set of literals of variables in $\formula$, $\wfunc:\ltr(\formula)\rightarrow \mathbb{Z}^+$ be a \emph{weight function} that assigns a (positive integer) weight value  to each literal, and $\mathcal{A}(\formula)$ be the set of all possible assignments to $\var(\formula)$.
    The \emph{weighted model count} $\WMC(\formula, \wfunc)$ of formula $\formula$ is defined as 
    \[
         \WMC(\formula, \wfunc):=
         \sum_{\substack{\model\in \mathcal{A}(\formula)\\ \model \models \formula}}
         \Big(
         \quad~
         \prod_{\mathclap{\substack{x\in \var(\formula)\\ \model(x)=1}}}\wfunc(x)
         \cdot
         \prod_{\mathclap{\substack{y\in \var(\formula)\\ \model(y)=0}}}\wfunc(\nl{y})
         \Big)
         .
    \]
\end{definition}

In the Weighted Model Counting problem (WMC), given a formula $\formula$ and a weight function $\wfunc$, the goal is to compute the weighted model count $\WMC(\formula, \wfunc)$ of formula $\formula$.
We use weighted \#2-SAT and weighted \#3-SAT to denote the restricted versions of WMC where the inputs are $2$-CNF and $3$-CNF formulas, respectively (the definition of 2-CNF and 3-CNF can be found below).

A clause containing a single literal $\literal$ may be simply written as $(\literal)$.
We use $C_1C_2$ to denote the clause obtained by concatenating clauses $C_1$ and $C_2$.
For a formula $\formula$, we denote $\formula[\literal=1]$ as the resulting formula obtained from $\formula$ by removing all clauses containing literal $\literal$ and removing all literals $\nl{\literal}$ from all clauses in $\formula$.

The \emph{degree} of a variable $x$ in formula $\formula$, denoted by $\deg(x)$, is the total number of occurrences of literals $x$ and $\nl{x}$ in $\formula$.
A \emph{$d$-variable} (resp., \emph{$d^+$-variable}) is a variable with degree exactly $d$ (resp., at least $d$). 
The degree of a formula $\formula$, denoted by $\deg(\formula)$, is the maximum degree of all variables in $\formula$.
The \emph{length} of a clause $C$ is the number of literals in $C$. 
A clause is a \emph{$k$-clause} (resp., \emph{$k^-$-clause}) if its length is exactly $k$ (resp., at most $k$).
A formula $\formula$ is called \emph{$k$-CNF formula} if each clause in $\formula$ has length at most $k$.

We say a clause $C$ \emph{contains} a variable $x$ if $x\in \var(C)$.
Two variables $x$ and $y$ are \emph{adjacent} (and \emph{neighbors} of each other) if they appear together in some clause.
We denote by $N(x, \formula)$ (resp., $N_i(x, \formula)$) the set of neighbors (resp., the set of $i$-degree neighbors) of variable $x$ in formula $\formula$.
When $\formula$ is clear from the context, we may simply write $N(x)$ (resp., $N_i(x)$).

\subsection{Graph-related Concepts}

The following two prominent graph representations of a CNF formula, namely the \emph{primal graph} and the \emph{dual graph}, will be used in our algorithms.
\begin{definition}[Primal graphs]
    The \emph{primal graph} $G(F)$ of a formula $F$ is a graph where each vertex corresponds to a variable in the formula.
    Two vertices $x$ and $y$ are adjacent if and only if $x, y\in \var(C)$ for some clause $C\in F$.
\end{definition}

\begin{definition}[Dual graphs]
    The \emph{dual graph} $G^d(F)$ of a formula $F$ is the graph where each vertex corresponds to a clause in the formula.
    Two vertices $C_1$ and $C_2$ are adjacent if and only if $\var(C_1)\cap \var(C_2)\neq\emptyset$ for  $C_1,C_2\in \formula$.    
\end{definition}

We also use the concepts of \textit{path decomposition}s, which offer a way to decompose a graph into a path structure.

\begin{definition}[Path decompositions]
    A \emph{path decomposition} of a graph $G$ is a sequence $P = (X_1, \dots, X_r)$ of vertex subsets $X_i \subseteq V(G)$ $(i \in \{1, \dots, r\})$ such that: (1) $\bigcup_{i = 1}^r X_i = V(G)$; (2) For every $uv \in E(G)$, there exists $l \in \{1, \dots, r\}$ such that $X_l$ contains both $u$ and $v$; (3) For every $u \in V(G)$, if $u \in X_i \cap X_k$ for some $i \leq k$, then $u \in X_j$  for all $i \leq j \leq k$.
\end{definition}
The \emph{width} of a path decomposition $(X_1, \dots, X_r)$ is defined as $\max_{1 \leq i \leq r} \{|X_i|\} - 1$.
The \emph{pathwidth} of a graph $G$, denoted by $\pw(G)$, is the minimum possible width of any path decomposition of $G$.
The \emph{primal pathwidth} and \emph{dual pathwidth} of a formula are the pathwidths of its primal graph and dual graph, respectively.

The following is a known bound in terms of the pathwidth.

\begin{theorem}[\citep{journals/algorithmica/FominGSS09}]\label{thm:pw-lowd}
    For any $\epsilon > 0$, there exists an integer $n_\epsilon$ such that for every graph $G$ with $n > n_\epsilon$ vertices,
    \[
        \pw(G) \leq n_3/6+n_4/3 + n_{\geq 5} + \epsilon n,
    \]
    where $n_i (i\in \{3, 4\})$ is the number of vertices of degree $i$ in $G$ and $n_{\geq 5}$ is the number of vertices of degree at least 5. Moreover, a path decomposition of the corresponding width can be constructed in polynomial time.
\end{theorem}

\citet{journals/jda/SamerS10} introduced fast algorithms for \#SAT parameterized by primal pathwidth and dual pathwidth. With minor modifications, these algorithms can be adapted to solve the weighted version, specifically WMC, without increasing the time complexity.

\begin{theorem}[\citep{journals/jda/SamerS10}]\label{thm:alg-primal-pw}
    Given an instance $(\formula, \wfunc)$ of WMC and a path decomposition $P$ of $G(\formula)$, 
    there is an algorithm (denoted by $\algpripw(\formula, \wfunc, P)$) that solves WMC in time $\bigOs(2^{p})$, where $p$ is the width of $P$.
\end{theorem}

\begin{theorem}[\citep{journals/jda/SamerS10}]\label{thm:alg-dual-pw}
    Given an instance $(\formula, \wfunc)$ of WMC and a path decomposition $P$ of $G^d(\formula)$, 
    there is an algorithm (denoted by $\algdualpw(\formula, \wfunc, P)$) that solves WMC in time $\bigOs(2^{p})$, where $p$ is the width of $P$.
\end{theorem}

\subsection{Branch-and-Search Algorithms}
A \emph{branch-and-search algorithm} first applies reduction rules to reduce the instance and then searches for a solution by branching.
We need to use a measure to evaluate the size of the search tree generated in the algorithm.
Let $\mu$ be the measure and $T(\mu)$ be an upper bound on the size of the search tree generated by the algorithm on any instance with the measure of at most $\mu$.
A branching operation, which branches on the instance into $l$ branches with the measure decreasing by at least $a_i>0$ in the $i$-th branch, is usually represented by a recurrence relation
\[
 T(\mu)\leq T(\mu-a_1)+\dots + T(\mu-a_l),
\]
or simply by a \emph{branching vector} $(a_1,\dots,a_l)$.
The \emph{branching factor} of the recurrence, denoted by $\tau(a_1,\dots,a_l)$, is the largest root of the function $f(x) = 1-\sum_{1\leq i\leq l} x^{-a_i}$.
If the maximum branching factor for all branching operations in the algorithm is at most $\gamma$, then $T(\mu)=O(\gamma^\mu)$.
More details about analyzing branching algorithms can be found in \citep{series/txtcs/FominK10}. We say that one branching vector is \emph{not worse} than the other if its corresponding branching factor is not greater than that of the latter.
The following useful property about branching vectors can be obtained from Lemma~2.2 and Lemma~2.3 in~\citep{series/txtcs/FominK10}.
\begin{lemma}\label{lem:vec-dominate}
 A branching vector $(a_1, a_2)$ is not worse than $(p-q, q)$ (or $(q, p-q)$) if $a_1+a_2\geq p$ and $a_1,a_2\geq q>0$.
\end{lemma}

\section{Framework of Algorithms}

An instance of WMC is denoted as $\ins=\shortins$.
For the sake of describing recursive algorithms,
we use $\ins=\longins$ to denote an instance, where $W$ is a positive integer and the solution to this instance is $W\cdot \WMC(F,w)$. Initially, it holds that $W=1$, which corresponds to the original WMC.

Our algorithms for weighted \#2-SAT and weighted \#3-SAT adopt the same framework, which contains three major phases. 
The \emph{first phase} is to apply some reduction rules to simplify the instance. 
The reduction rules we use in the algorithm will be introduced in the next subsection.

The \emph{second phase} is to branch on some variable by assigning either $1$ or $0$ to it. 
This phase will create branching vectors and exponentially increase the running time of the algorithm.
Specifically, branching on variable $x$ in an algorithm \texttt{Alg} means doing the following:
\begin{itemize}
    \item  $W_t\gets \texttt{Alg}(\formula[x=1], \wfunc, \solW)$;
    \item $W_f\gets \texttt{Alg}(\formula[x=0], \wfunc, \solW)$;
    \item Return $w(x)\cdot W_t + w(\nl{x})\cdot W_f$.
\end{itemize}

In our algorithms, we may only branch on variables of (relatively) high degree.
When all variables have a low degree, the corresponding primal or dual graphs usually have a small pathwidth. 
In this case, we will apply Theorem~\ref{thm:pw-lowd} to obtain a path decomposition with small width, and then invoke the algorithms in Theorem~\ref{thm:alg-primal-pw} and Theorem~\ref{thm:alg-dual-pw} to solve the problem. 
This is the \emph{third phase} of our algorithms.
Before introducing our algorithms for weighted \#2-SAT and weighted \#3-SAT, we first introduce our reduction rules for general WMC.
Since our reduction rules are applicable for general WMC, they can also be applied to both weighted \#2-SAT and weighted \#3-SAT.

\subsection{Reduction Rules}
A reduction rule takes $\ins=\longins$ as input and outputs a new instance $\ins'=(\formula', \wfunc', \solW')$.
A reduction rule is \emph{correct} if $\WMC(\formula, \wfunc)\cdot\solW = \WMC(\formula', \wfunc')\cdot\solW'$ holds.

In total, we have nine reduction rules. 
When we consider a rule, we assume that all previous rules are not applicable.

The first four reduction rules are simple and well-known.
\begin{rrule}[Elimination of duplicated literals]\label{rrule:duplicated}
    If a clause $C$ contains duplicated literals $\literal$, remove all but one $\literal$ in $C$.
\end{rrule}

\begin{rrule}[Elimination of tautology]\label{rrule:taut}
    If a clause $C$ contains two complementary literals $\literal$ and $\nl{\literal}$, remove clause $C$.
\end{rrule}

\begin{rrule}[Elimination of subsumptions]\label{rrule:sub}
    If there are two clauses $C$ and $D$ such that $C\subseteq D$, remove clause $D$.
\end{rrule}

\begin{rrule}[Elimination of $1$-clauses]\label{rrule:1cls}
    If there is a $1$-clause $(\literal)$, then $W\gets W\cdot w(\literal)$, and $\formula \gets \formula[\literal=1]$.
\end{rrule}

In the algorithms, we may also generate $0$-variables that are unassigned yet. The following rule can eliminate them. 
\begin{rrule}[Elimination of $0$-variables]\label{rrule:0var}
    If there is a unassigned variable $x$ with $\deg(x)=0$, let $\solW \gets \solW \cdot (w(x)+w(\nl{x}))$ and remove variable $x$.
\end{rrule}

Next two rules are going to deal with some $2$-clauses.
\begin{rrule}\label{rrule:2cls-1comp}
    If there are two clauses $\literal_{a} \literal_{b}$ and $\literal_{a}\nl{\literal_{b}}C$ in $F$, then remove literal $\nl{\literal_{b}}$ from clause $\literal_{a}\nl{\literal_{b}}C$.
\end{rrule}

\begin{rrule}\label{rrule:2cls-2comp}
    If there are two clauses $\literal_{a} \literal_{b}$ and $\nl{\literal_{a}}\nl{\literal_{b}}$ in $F$, do the following:
    \begin{enumerate}
        \item $\wfunc(\nl{\literal_{b}})\gets \wfunc(\nl{\literal_{b}})\cdot \wfunc(\literal_{a})$, 
                $\wfunc(\literal_{b})\gets \wfunc(\literal_{b})\cdot \wfunc(\nl{\literal_{a}})$;
        \item replace $\literal_{a}$ (resp., $\nl{\literal_{a}}$) with $\nl{\literal_{b}}$ (resp., $\literal_{b}$) in $F$;
        \item remove $\var(\literal_{a})$ and apply R-Rule~\ref{rrule:taut} as often as possible. 
    \end{enumerate}
      
\end{rrule}

We use $\algbf\shortins$ to denote the brute-force algorithm that solves WMC by enumerating all possible assignments.
Clearly, $\algbf\shortins$ runs in $\bigOs(2^{n(\formula)})$ time
and it uses constant time if the formula has only a constant number of variables. The following two rules are based on a divide-and-conquer idea. However, we only apply them to the cases where one part is of constant size. 

\begin{rrule}\label{rrule:disjoint}
    If formula $\formula$ can be partitioned into two non-empty sub-formulas $F_1$ and $F_2$ with $\var(F_1)\cap \var(F_2)=\emptyset$ and $n(F_1)\leq 10$, do the following:
    \begin{enumerate}
        \item $\solW'\gets \algbf(F_1, \wfunc)$;
        \item $\solW\gets \solW\cdot \solW'$, and $\formula\gets F_2$.
    \end{enumerate}
\end{rrule}

\begin{rrule}\label{rrule:cut}
    If there is a variable $x$ such that formula $\formula$ can be partitioned into two non-empty sub-formulas $F_1$ and $F_2$, with $\var(F_1)\cap \var(F_2)=\{x\}$ and $n(F_1)\leq 10$, do the following:
    \begin{enumerate}
        \item $\solW_t\gets \algbf(F_1[x=1], \wfunc)$;
        \item $\solW_f\gets \algbf(F_1[x=0], \wfunc)$;
        \item $w(x)\gets w(x)\cdot \solW_t$, and $w(\nl{x})\gets w(\nl{x})\cdot \solW_f$;
        \item $\formula \gets F_2$.
    \end{enumerate}
\end{rrule}

\begin{lemma}\label{lem:reduction-correct}
    All of the above nine reduction rules are correct.
\end{lemma}

\begin{proof}
For convenience, we extend an assignment $\sigma$ from variables to literals by setting $\sigma(\nl{x})=1-\sigma(x)$ for any variable $x$.
According to our definition of correctness, we need to prove that for each instance $\mathcal{I} = (F, w, W)$, the instance $\mathcal{I}' = (F', w', W')$ output by each R-Rule satisfies $\WMC(F, w) \cdot W = \WMC(F', w') \cdot W'$. The correctness of the first six rules is relatively straightforward, and we only provide brief explanations here. We will give detailed proofs for the last three rules. Note that during the execution of R-Rule~\ref{rrule:duplicated}, R-Rule~\ref{rrule:taut}, R-Rule~\ref{rrule:sub}, and R-Rule~\ref{rrule:2cls-1comp}, both $w$ and $W$ remain unchanged. For these rules, we only need to show that an assignment satisfies the original formula if and only if it satisfies the new formula.

    \textbf{R-Rule~\ref{rrule:duplicated}:} A literal and its duplicates in a clause $C$ are assigned the same value by any assignment $\sigma$. Thus, removing the duplicates does not change the satisfiability of $C$ under $\sigma$. The set of satisfying assignments for $F$ remains unchanged.

    \textbf{R-Rule~\ref{rrule:taut}:} Any assignment $\sigma$ over $\var(F)$ satisfies a clause $C$ that contains a literal and its negation. Thus, removing $C$ does not affect the set of satisfying assignments for the formula.

    \textbf{R-Rule~\ref{rrule:sub}:} If $C \subseteq D$, any assignment $\sigma$ that satisfies $C$ must also satisfy $D$. Therefore, clause $D$ is redundant if clause $C$ must be satisfied. Removing $D$ does not change the set of satisfying assignments for the formula.

    \textbf{R-Rule~\ref{rrule:1cls}:} Any satisfying assignment for $F$ must assign the literal $\literal$ to $1$. The total weight of satisfying assignments for $F$ is the sum of weights of assignments that satisfy both $\literal=1$ and the resulting formula $F[\literal=1]$. This can be expressed as $\WMC(F, w) = w(\literal) \cdot \WMC(F[\literal=1], w)$. The rule sets $F' = F[\literal=1]$ and $W' = W \cdot w(\literal)$, while $w'=w$. Thus, $\WMC(F, w) \cdot W = (w(\literal) \cdot \WMC(F', w')) \cdot W = \WMC(F', w') \cdot W'$.

    \textbf{R-Rule~\ref{rrule:0var}:} A variable $x$ with $\deg(x)=0$ does not appear in any clause. Thus, its value does not affect the satisfiability of $F$. The total weighted model count can be factored as $\WMC(F, w) = (w(x) + w(\nl x)) \cdot \WMC(F', w)$, where $F'$ is the resulting formula after removing variable $x$ from the variable set of formula $F$. The rule sets $W' = W \cdot (w(x) + w(\nl x))$ and $F'$ as described. Thus, the correctness of this rule holds.

    \textbf{R-Rule~\ref{rrule:2cls-1comp}:} Let the two clauses be $C_1 = (\literal_a \lor \literal_b)$ and $C_2 = (\literal_a \lor \nl{\literal_b} \lor C)$. The rule transforms $C_2$ into $C_2' = (\literal_a \lor C)$. We show that for any assignment $\sigma$, it holds that $\sigma \models (C_1 \land C_2)$ if and only if $\sigma \models (C_1 \land C_2')$ by considering the following two cases.
    \begin{itemize}
        \item If $\sigma(\literal_a) = 1$, both $C_1, C_2, C_2'$ are satisfied. The equivalence holds.
        \item If $\sigma(\literal_a) = 0$, for $\sigma$ to satisfy $(C_1 \land C_2)$, it must satisfy $(\literal_b)$ and $(\nl{\literal_b} \lor C)$, which implies $\sigma(\literal_b)=1$ and $\sigma \models C$. For $\sigma$ to satisfy $(C_1 \land C_2')$, it must satisfy $(\literal_b)$ and $(C)$, which also implies $\sigma(\literal_b)=1$ and $\sigma \models C$. The equivalence holds.
    \end{itemize}
    Since for any assignment, it satisfies the original formula if and only if it satisfies the new formula, the rule is correct.

    \textbf{R-Rule~\ref{rrule:2cls-2comp}:} The presence of clauses $(\literal_a \lor \literal_b)$ and $(\nl{\literal_a} \lor \nl{\literal_b})$ implies that any satisfying assignment $\sigma$ for $F$ must have $\sigma(\literal_a) \neq \sigma(\literal_b)$, which is equivalent to $\sigma(\literal_a) = \sigma(\nl{\literal_b})$. We establish a bijection between the satisfying assignments of $F$ and $F'$.
    For any satisfying assignment $\sigma$ of $F$, we define an assignment $\sigma'$ for $F'$ by setting $\sigma'(y) = \sigma(y)$ for all $y \in \var(F')$. Since $\var(\literal_a)$ is removed and all its occurrences are replaced, $\sigma'$ is well-defined.
    The weight of a satisfying assignment $\sigma$ for $F$ can be factored based on the value of $\literal_a$.
    \begin{itemize}
        \item If $\sigma(\literal_a)=1$, then $\sigma(\literal_b)=0$. The weight contribution from these two variables is $w(\literal_a) \cdot w(\nl{\literal_b})$. In $F'$, $\literal_a$ is replaced by $\nl{\literal_b}$. The corresponding assignment $\sigma'$ has $\sigma'(\nl{\literal_b})=1$. The new weight is $w'(\nl{\literal_b}) = w(\literal_a) \cdot w(\nl{\literal_b})$, matching the contribution.
        \item If $\sigma(\literal_a)=0$, then $\sigma(\literal_b)=1$. The weight contribution is $w(\nl{\literal_a}) \cdot w(\literal_b)$. In $F'$, $\nl{\literal_a}$ is replaced by $\literal_b$. The corresponding assignment $\sigma'$ has $\sigma'(\literal_b)=1$. The new weight is $w'(\literal_b) = w(\nl{\literal_a}) \cdot w(\literal_b)$, matching the contribution.
    \end{itemize}
    This establishes a weight-preserving bijection between the satisfying assignments of $F$ and $F'$. Thus, $\WMC(F, w) = \WMC(F', w')$, and the rule is correct.

    \textbf{R-Rule~\ref{rrule:disjoint}:} Since $\var(F_1) \cap \var(F_2) = \emptyset$, the set of satisfying assignments of $F$ is the Cartesian product of the satisfying assignments of $F_1$ and $F_2$. Due to the multiplicative nature of weights, the WMC of $F$ is the product of the WMCs of its components:
    $$ \WMC(F, w) = \WMC(F_1, w) \cdot \WMC(F_2, w). $$
    The rule computes $\solW' = \WMC(F_1, w)$ (in constant time as $n(F_1) \le 10$), sets $F' = F_2$, and updates the global multiplier to $W_{new} = W_{old} \cdot \solW'$.
    We have $W_{old} \cdot \WMC(F, w) = W_{old} \cdot (\WMC(F_1, w) \cdot \WMC(F_2, w)) = (W_{old} \cdot \solW') \cdot \WMC(F_2, w) = W_{new} \cdot \WMC(F', w')$. Thus, the rule is correct.

    \textbf{R-Rule~\ref{rrule:cut}:} The total WMC of $F$ can be computed by conditioning on the value of the variable $x$:
    $$ \WMC(F, w) = w(x) \cdot \WMC(F[x=1], w) + w(\nl{x}) \cdot \WMC(F[x=0], w). $$
    Since $\var(F_1)\cap \var(F_2)=\{x\}$, we have $\WMC(F[x=1], w) = \WMC(F_1[x=1], w) \cdot \WMC(F_2[x=1], w)$ and $\WMC(F[x=0], w) = \WMC(F_1[x=0], w) \cdot \WMC(F_2[x=0], w)$.
    Let $W_t = \WMC(F_1[x=1], w)$ and $W_f = \WMC(F_1[x=0], w)$, which are computed in constant time. Then,
    $$ \WMC(F, w) = w(x) \cdot W_t \cdot \WMC(F_2[x=1], w) + w(\nl{x}) \cdot W_f \cdot \WMC(F_2[x=0], w). $$
    The rule sets $F' = F_2$ and defines a new weight function $w'$ where $w'(x) = w(x) \cdot W_t$, $w'(\nl{x}) = w(\nl{x}) \cdot W_f$, and $w'(y) = w(y)$ for other literals $y$. The WMC of the new instance is:
    \begin{align*}
        \WMC(F', w') 
        &= w'(x) \cdot \WMC(F_2[x=1], w) + w'(\nl{x}) \cdot \WMC(F_2[x=0], w)\\
        &= (w(x) \cdot W_t) \cdot \WMC(F_2[x=1], w) + (w(\nl{x}) \cdot W_f) \cdot \WMC(F_2[x=0], w).
    \end{align*}
    Thus, we have $\WMC(F, w) = \WMC(F', w')$. Since $W$ is unchanged ($W'=W$), the rule is correct.
\end{proof}

\begin{definition}[Reduced formulas]
    A formula $\formula$ is called  \textit{reduced} if none of the above reduction rules is applicable.
    We use $\reduced{\formula}$ to denote the reduced formula obtained by exhaustively applying the above reduction rules on $\formula$.
\end{definition}

\begin{lemma}\label{lem:reduction-safe}
    For any formula $\formula$, applying any reduction rule on $\formula$ will not increase the number of clauses or the length of any clause.
    Moreover, it takes polynomial time to transform $\formula$ into a reduced formula $\reduced{\formula}$ where no reduction rules can be applied.
\end{lemma}
\begin{proof}
    First, we verify that no reduction rule increases the number of variables $n(\formula)$, the number of clauses $m(\formula)$, or the total length of all clauses $L(\formula)$. 
    It can be verified that none of the nine rules introduce new variables. 
    R-Rule~\ref{rrule:duplicated}, R-Rule~\ref{rrule:2cls-1comp}, and R-Rule~\ref{rrule:2cls-2comp} only remove or replace literals; thus, they do not increase $m(\formula)$ or $L(\formula)$. The other rules (R-Rule~\ref{rrule:taut}, R-Rule~\ref{rrule:sub}, R-Rule~\ref{rrule:1cls}, R-Rule~\ref{rrule:0var}, R-Rule~\ref{rrule:disjoint}, and R-Rule~\ref{rrule:cut}) simply remove clauses or variables; so they do not increase these parameters either. Thus, applying any reduction rule does not increase the number of clauses or the length of any clause.

    To prove that the reduction process terminates in polynomial time, we define a potential function $\Phi(\formula) = n(\formula) + m(\formula) + L(\formula)$. The initial value of $\Phi(\formula)$ is polynomial in the input size.
    We show that each application of any rule decreases $\Phi(\formula)$ by at least 1.
    
    Previous analysis shows that the application of any rule would not increase $n(\formula)$, $m(\formula)$, or $L(\formula)$. 
    Thus, it is sufficient to show that each application strictly decreases at least one of $n(\formula)$, $m(\formula)$, and $L(\formula)$: 
    each application of Rule~\ref{rrule:duplicated} and R-Rule~\ref{rrule:2cls-1comp} removes at least one literal, which decreases $L(\formula)$ by at least 1; 
    applying one of R-Rule~\ref{rrule:taut}, R-Rule~\ref{rrule:sub}, and R-Rule~\ref{rrule:1cls} removes at least one clause, which decreases $m(\formula)$ by at least 1;
    each application of R-Rule~\ref{rrule:0var} and R-Rule~\ref{rrule:2cls-2comp} removes at least one variable, which decreases $n(\formula)$ by at least 1;
    R-Rule~\ref{rrule:disjoint} and R-Rule~\ref{rrule:cut} remove a non-empty sub-formula, which involves removing at least one variable and one clause and decreases both $n(\formula)$ and $m(\formula)$ by at least 1.

    Since each application of a reduction rule decreases $\Phi(\formula)$ by at least 1 and $\Phi(\formula) \ge 0$, the reduction process must terminate. The total number of applications of reduction rules is bounded by the initial value of $\Phi(\formula)$, which is polynomial in the input size. Furthermore, each rule can be checked and applied in polynomial time. 
    Therefore, transforming $\formula$ into $\reduced{\formula}$ can be done in polynomial time.
\end{proof}

\begin{lemma}\label{lem:reduced_prop_basic}
    In a reduced formula $\formula$, it holds that (1) all clauses are $2^+$-clauses; (2) all $2$-clauses only contains $2^+$-variables.
\end{lemma}
\begin{proof}
    By R-Rule~\ref{rrule:1cls}, there is no $1$-clause in $\formula$, and so (1) holds.
    By R-Rule~\ref{rrule:0var}, there is no $0$-variable in $\formula$.
    Let $C$ be a $2$-clause in $\formula$.
    We partition formula $\formula$ into $F_1=C$ and $F_2=\formula\setminus F_1$.
    If $C$ contains two $1$-variables, then R-Rule~\ref{rrule:disjoint} can be applied since $\var(F_1)\cap \var(F_2) = \emptyset$.
    If $C$ contains only one $1$-variable (assume that $x$ the $2$-variable in $C$), then R-Rule~\ref{rrule:cut} is applicable since $\var(F_1)\cap \var(F_2) = \{x\}$. Thus, $C$ would not contain $1$-variables. This completes the proof.
\end{proof}

\begin{lemma}\label{lem:reduced_prop_neighbor}
    In a reduced formula $\formula$, if there is a $2$-clause $\literal_{a}\literal_{b}$, there is no other clause containing $\literal_{a}\literal_{b}$, $\nl{\literal_{a}}\literal_{b}$, or ${\literal_{a}}\nl{\literal_{b}}$, and there is no $2$-clause $\nl{\literal_{a}}\nl{\literal_{b}}$.
\end{lemma}
\begin{proof}
    Suppose, for the sake of contradiction, that there is a $2$-clause $\literal_{a}\literal_{b}$ and there exists a clause $C$ which contradicts what the lemma says. We consider case by case.
    (1) If $C$ contains $\literal_{a}\literal_{b}$, R-Rule~\ref{rrule:taut} is applicable.
    (2) If $C$ contains $\nl{\literal_{a}}\literal_{b}$ (or ${\literal_{a}}\nl{\literal_{b}}$), R-Rule~\ref{rrule:2cls-1comp} can be applied.
    (3) If $C = \nl{\literal_{a}}\nl{\literal_{b}}$, R-Rule~\ref{rrule:2cls-2comp} is applicable.
    In conclusion, we can apply reduction rules in any case, a contradiction. Thus, the lemma holds.
\end{proof}

\section{The Algorithm for Weighted \#2-SAT}
In this section, we introduce our algorithm, called $\algtwocnf$, for WMC on $2$-CNF formulas.
The algorithm is presented in Algorithm~\ref{alg:2CNF}.
As we mentioned before, the algorithm comprises three main phases. Phase one (Line 1) is to apply reduction rules to get a reduced instance.
Phase two (Lines 4--8) is going to branch on $5^+$-variables and some special $4$-variables.
After phase two, the primal graph of the formula admits a small pathwidth. 
Phase three (Steps 10-11) is, based on a path decomposition, to use the algorithm $\algpripw$ in Theorem~\ref{thm:alg-primal-pw} to solve the problem directly.

\begin{algorithm}[htbp]
    \caption{$\algtwocnf\longins$}
    \label{alg:2CNF}
    \textbf{Input}: 2-CNF formula $\formula$, weight function $w$, and integer $\solW$.\\
    \textbf{Output}: The weighted model count $\solW\cdot \WMC(\formula, \wfunc)$.
    
    \begin{algorithmic}[1] %[1] enables line numbers
        \State Apply reduction rules exhaustively to reduce $\formula$ and update $\wfunc$ and $\solW$ accordingly.
        \IfSingleLine{$\formula$ is empty}{
            \Return $\solW$.
        }
        \IfSingleLine{$\formula$ contains empty clause}{
            \Return $0$.
        }
        \If{$\deg(\formula)\geq 5$}
            \State Select a variable $x$ with $\deg(x)=\deg(\formula)$;
            \State Branch on $x$.\label{line:alg2-branch-1} %
        \ElsIf{$\exists$ $x$ such that $\deg(x)=4$ and $|N_4(x)|\geq 3$}
            \State Branch on $x$.\label{line:alg2-branch-2} %
        \Else 
            \State $P\gets$ path decomposition of $G(\formula)$ via Theorem~\ref{thm:pw-lowd}. \label{line:alg2-pd}
            \State $\solW_{pw}\gets \algpripw(\formula, \wfunc, P)$; \label{line:alg2-call-pw} %
            \State \Return $\solW \cdot \solW_{pw}$.
        \EndIf
    \end{algorithmic}
    
\end{algorithm}

\subsection{The Analysis}
Although the algorithm itself is simple, its running time analysis is technically involved.
We first prove some properties of a reduced 2-CNF, which will be used in our analysis.

\begin{lemma}\label{lem:2cnf-basic-prop}
    In a reduced $2$-CNF formula $\formula$, it holds that (1) all clauses are $2$-clauses; (2) all variables are $2^+$-variables; (3) $n(\formula)\leq m(\formula)$.
\end{lemma}
\begin{proof}
    Since $\formula$ is a $2$-CNF formula, (1) and (2) immediately follow from Lemma~\ref{lem:reduced_prop_basic}.
    Let $L(\formula)$ be the sum of the lengths of clauses in $\formula$.
    Since all clauses are $2$-clauses, we have $L(\formula)=2m(\formula)$.
    On the other hand, since all variables are $2^+$-variables, we have $L(\formula)\geq 2n(\formula)$.
    Thus, $m(\formula)\geq n(\formula)$ holds.
    This completes the proof.
\end{proof}

\begin{lemma}\label{lem:2cnf-N2-sep}
    In a reduced $2$-CNF formula $\formula$, for a variable $x$, any clause contains at most one variable in $N_2(x)$.
\end{lemma}
\begin{proof}
    Suppose for contradiction that there is a $2$-clause $C_{yz}$ which contains two variables $y$ and $z$ in $N_2(x)$ for some variable $x$.
    Let $D_{xy}$ (resp.,$D_{xz}$) be the $2$-clause containing variables $x$ and $y$ (resp., variables $x$ and $z$).
    Since both $y$ and $z$ are $2$-variables, formula $\formula$ can be partitioned into $F_1$ and $F_2$, where $F_1 = C_{xy}\wedge C_{xz}\wedge C_{yz}$ and $F_2 = \formula\setminus F_1$, such that $\var(F_1)\cap \var(F_2) = \{x\}$ and $n(F_1)=3$.
    Thus, R-Rule~\ref{rrule:cut} is applicable, which contradicts that $\formula$ is reduced.
\end{proof}

\begin{lemma}\label{lem:2cnf-assign-basic}
    Let $\formula$ be a reduced $2$-CNF formula and $x$ be a variable in $\formula$.
    All clauses containing $x$ would not appear in $\reduced{\formula[x=0]}$ and $\reduced{\formula[x=1]}$.
\end{lemma}
\begin{proof}
    By Lemma~\ref{lem:2cnf-basic-prop}, all clauses in $\formula$ are $2$-clauses.
    Consider the case that we assign $x=1$, and the case for $x=0$ is analogous.
    All clauses that contain literal $x$ are satisfied and removed.
    Furthermore, all clauses that contain literal $\nl{x}$ become $1$-clauses, and thus R-Rule~\ref{rrule:1cls} would be applied to remove them.
    Thus, all clauses containing $x$ would not appear in $\reduced{\formula[x=0]}$ and $\reduced{\formula[x=1]}$.
\end{proof}

To analyze the running time bound, we focus on phase two and phase three since phase one will not exponentially increase the running time. For phase two, we mainly use Lemma~\ref{lem:vec-dominate} to get the worst branching vector. Thus, we need to analyze lower bounds for the decrease of $m(\formula)$ in a branching operation, which is formally presented in Lemma~\ref{lem:2cnf-general-bound} below.

\begin{lemma}\label{lem:2cnf-general-bound}
    Let $\formula$ be a reduced 2-CNF formula of degree $d$ and $x$ be a $d$-variable in $\formula$.
    Let $\Delta_{t}=m(\formula)-m(\reduced{\formula[x=1]})$ and $\Delta_{f}=m(\formula)-m(\reduced{\formula[x=0]})$.
    It holds that
    \begin{enumerate}
        \item $\Delta_{t}, \Delta_{f}\geq d + |N_2(x)|$;
        \item $\Delta_{t}+\Delta_{f}\geq 2d + \card{N_2(x)} + \left\lceil \frac{1}{2}\sum_{2\leq i\leq d}(i-1)|{N_i(x)}| \right\rceil + 1$ if $d\leq 7$.
    \end{enumerate}
\end{lemma}
\begin{proof}
    For clarity, we define the following notations:
    \begin{itemize}%[itemsep=0pt, topsep=0pt]
        \item $S_0$: the set of clauses that contain variable $x$.
        \item $S_1$: the set of clauses that contain variable(s) in $N_{2}(x)$ but not contain variable $x$.
        \item $S_2$: the set of clauses that contain variable(s) in $N_{i}(x)$, where $i\geq 3$, but not contain any variable in $N_{2}(x)\cup \{x\}$.
    \end{itemize}

    \begin{figure}
        \centering
        \includegraphics[width=0.7\textwidth]{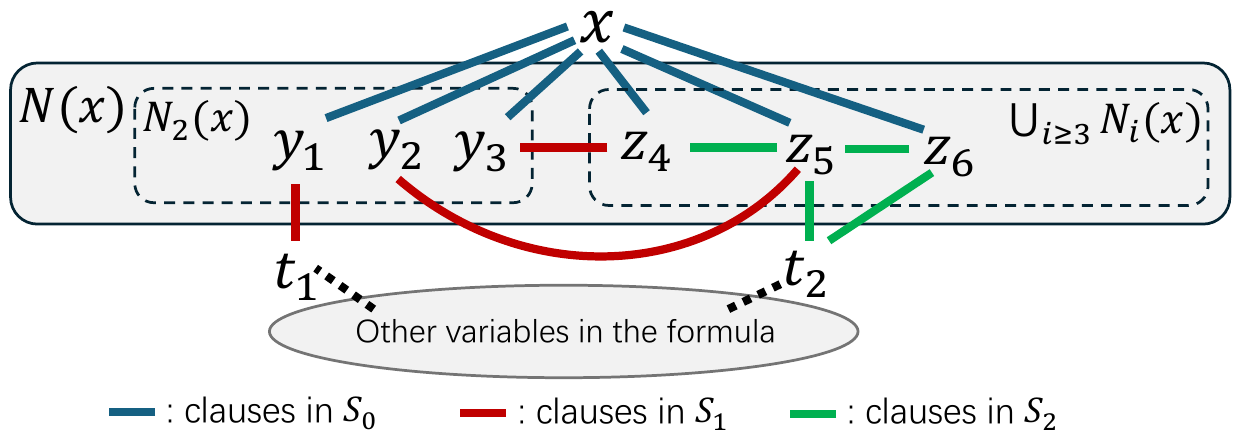}
        \caption{An illustrative example of the primal graph of a formula, highlighting the relationships among variables $x$, variables in $N_2(x)=\{y_1,y_2,y_3\}$ and $\bigcup_{i\geq 3}N_{i}(x)=\{z_4,z_5,z_6\}$, and variables $t_1, t_2$ (that co-occur in clauses with $N(x)$).
        By Lemma~\ref{lem:2cnf-basic-prop}, each edge in the graph corresponds to a unique clause in $\formula$.
        By Lemma~\ref{lem:2cnf-N2-sep}, there is no edge between variables in $N_2(x)$.
        }
        \label{fig:illustration}
    \end{figure}

    By definition, $S_0, S_1, S_2$ are pairwise disjoint.
    The primal graph of $\formula$ shown in Figure~\ref{fig:illustration} provides a useful perspective for understanding these notations and the subsequent proofs.

    We plan to analyze the bounds for $\Delta_{t}$, $\Delta_{f}$, and $\Delta_{t} + \Delta_{f}$ by considering whether a clause in $S_0\cup S_1\cup S_2$ would be removed after we assign a value to variable $x$ (\ie, whether a clause would appear in $\reduced{\formula[x=0]}$ or $\reduced{\formula[x=0]}$).

    By Lemma~\ref{lem:2cnf-assign-basic}, all clauses in $S_0$ would not appear in both $\reduced{\formula[x=0]}$ and $\reduced{\formula[x=1]}$.
    
    We claim that all clauses in $S_1$ would not appear in both $\reduced{\formula[x=0]}$ and $\reduced{\formula[x=1]}$.
    Let $C$ be a clause containing variable $x$ and a variable $y\in N_{2}(x)$.
    After we assign a value to $x$, if clause $C$ is satisfied, then it is removed and $y$ becomes a $1$-variable.
    Then, since $y$ (as a $1$-variable) is contained in a $2$-clause, R-Rule~\ref{rrule:cut} can be applied to remove the other clause containing $y$.
    Otherwise, if clause $C$ is not satisfied, the literal of $x$ is removed from $C$ and so $C$ becomes a $1$-clause. In this case, applying R-Rule~\ref{rrule:1cls} would assign a value to variable $y$. By Lemma~\ref{lem:2cnf-assign-basic}, all clauses containing $y$ would be removed.
    Thus, all clauses that contain some variable in $N_{2}(x)$ (which includes $S_1$) would not appear in both $\reduced{\formula[x=0]}$ and $\reduced{\formula[x=1]}$.

    Since all clauses in $S_0$ and $S_1$ would not appear in both $\reduced{\formula[x=0]}$ and $\reduced{\formula[x=1]}$, we have 
    $\Delta_{t}\geq \card{S_0}+\card{S_1}$ and $\Delta_{f}\geq \card{S_0}+\card{S_1}$.
    Since $\card{S_0}=d$ (by definition) and $\card{S_1}=\card{N_{2}(x)}$ (by Lemma~\ref{lem:2cnf-N2-sep}), 
    it holds that $\Delta_t, \Delta_f\geq d + \card{N_{2}(x)}$.

    Next, we consider the clauses in $S_2$.
    Let $D$ be a clause that contains variable $x$ and a variable $z\in N_{i}(x)$ where $i\geq 3$.
    Assigning either $x=0$ or $x=1$ would make $D$ become a $1$-clause that only contains variable $z$.
    When clause $D$ becomes such $1$-clause, R-Rule~\ref{rrule:1cls} would be applied to assign a value to $z$, and then remove all clauses containing $z$ according to Lemma~\ref{lem:2cnf-assign-basic}.
    Thus, for a clause in $S_2$, it would not appear in at least one of $\reduced{\formula[x=0]}$ and $\reduced{\formula[x=1]}$.
    Together with previous analyses on $S_0$ and $S_1$, we have 
    \begin{equation}\label{eq:bound-S2-left}
    \begin{aligned}
        \Delta_{t} + \Delta_{f}
        \geq 2\card{S_0} + 2\card{S_1} + \card{S_2}
        = 2d + 2\card{N_{2}(x)} + \card{S_2}.
    \end{aligned}
    \end{equation}

    To accurately characterize $S_2$ (and $\card{S_2}$), we need some additional notations.
    For a variable $y\in N(x)$, we define $P(y):=N(y)\cap N(x)$ and $Q(y):=N(y)\setminus (N(x)\cup \{x\})$.
    For example, in the example shown in Figure~\ref{fig:illustration}, we have $P(z_5)=\{y_2,z_4,z_6\}$ and $Q(z_5)=\{t_2\}$.
    For each $i\geq 2$, we define $p_{i}:=\sum_{y\in N_{i}(x)}\card{P(y)}$ and $q_{i}:=\sum_{y\in N_{i}(x)}\card{Q(y)}$.

    Since all clauses are $2$-clauses, $\card{N(y)}=\deg(y)$ holds from Lemma~\ref{lem:reduced_prop_neighbor}.
    With this, we have the following that would come in handy later:
    \begin{align*}
        p_{i} + q_{i} 
        &=  \sum_{y\in N_{i}(x)}\card{P(y)}+\card{Q(y)} &\\
        &=  \sum_{y\in N_{i}(x)}\card{N(y)\setminus\{x\}} & \text{by $N(y)=\{x\}\sqcup P(y)\sqcup Q(y)$}\\
        &= \sum_{y\in N_{i}(x)}(\deg(y)-1) & \text{by $\card{N(y)}=\deg(y)$}\\
        &= \sum_{\mathclap{y\in N_{i}(x)}}(i-1) 
        = (i-1)|{N_{i}(x)}|.
    \end{align*}
    We write $p_{\geq 3}:=\sum_{i\geq 3}p_{i}$ and $q_{\geq 3}:= \sum_{i\geq 3}q_{i}$ for brevity. 
    We claim that
        \begin{equation}\label{eq:S_2}
            \card{S_2}=\frac{1}{2}(p_{\geq 3} - p_2) + q_{\geq 3} \mbox{ and } (p_{\geq 3} - p_2)\bmod{2}=0.
        \end{equation}
    The reason is as follows. 
    By definition, $\card{S_2}$ is the number of clauses that contain variable(s) in $\bigcup_{i\geq 3}N_{i}(x)$, but not contain any variable in $N_{2}(x)\cup \{x\}$. We may simply write $N_{\geq 3}(x):= \bigcup_{i\geq 3}N_{i}(x)$.
    Consider the primal graph shown in Figure~\ref{fig:illustration}.
    $\card{S_2}$ is the number of ``green edges'' in Figure~\ref{fig:illustration}.
    We further partition $S_{3}$ into $S_{3}^{in}\sqcup S_{3}^{out}$, where $S_{3}^{in}$ is the set of clauses that contain two variables in $\bigcup_{i\geq 3}N_{i}(x)$ (e.g., $z_4z_5$ and $z_5z_6$ in the figure) and $S_{3}^{out}$ is the set of clauses that contain a variable in $N_{\geq 3}(x)$ (e.g., $z_5t_2$ and $z_6t_2$ in the figure).
    We have $q_{\geq 3} = \card{S_{3}^{out}}$ by definition.
    Since $p_2$ is the number of clauses that contain a variable in $N_2(x)$ and a variable in $N_{\geq 3}(x)$ (\ie, number of red edges between $N_2(x)$ and $N_{\geq 3}(x)$),
    We have $p_{\geq 3} = 2\card{S_{3}^{in}} + p_2$.
    This gives us $\card{S_{3}^{in}} = \frac{1}{2}(p_{\geq 3} - p_2)$ and $(p_{\geq 3} - p_2)\bmod{2}=0$.
    With $q_{\geq 3} = \card{S_{3}^{out}}$, we have $\card{S_2}=\card{S_{3}^{in}}+\card{S_{3}^{out}}=\frac{1}{2}(p_{\geq 3} - p_2)+q_{\geq 3}$.
    
    By putting \eqref{eq:S_2} into \eqref{eq:bound-S2-left} and writing $\beta:= (p_{\geq 3} - p_2) + 2(p_2 + q_2) + 2q_{\geq 3}$ for convenience, we have
    \begin{equation}\label{eq:t+f}
    \begin{aligned}
        \Delta_{t} + \Delta_{f} 
        &\geq 2d + 2(p_{2}+q_{2}) + \frac{1}{2}(p_{\geq 3} - p_2) + q_{\geq 3}\\
        &= 2d+(p_2+q_2)+\frac{1}{2}\beta
        = 2d+\card{N_2(x)}+\frac{1}{2}\beta.
    \end{aligned}
    \end{equation}
    Next, we claim that 
    \begin{equation}\label{eq:beta-bound}
    \begin{aligned}
            \frac{1}{2}\beta \geq \left\lceil\frac{1}{2}\sum_{2\leq i\leq d}(i-1)|N_i(x)|\right\rceil + 1.
    \end{aligned}
    \end{equation}
    Note that with Eq.~\eqref{eq:t+f}, the lemma directly follows from Eq.~\eqref{eq:beta-bound}. 
    Next, we only need to prove Eq.~\eqref{eq:beta-bound}.
    
    By $p_{i} + q_{i} = (i-1)|{N_{i}(x)}|$, we have
    \begin{equation}\label{eq:beta}
    \begin{aligned}
        \beta
        & = (p_{\geq 3} - p_2) + 2(p_2 + q_2) + 2q_{\geq 3}\\
        & = (p_{\geq 3} + p_2 + 2q_2 + 2q_{\geq 3})\\
        & = \sum_{2\leq i\leq d}(p_i+q_i) + q_{\geq 3} + q_2\\
        & = \sum_{2\leq i\leq d}(i-1)|N_i(x)| + q_{\geq 3} + q_2.
    \end{aligned}
    \end{equation}

    We need some discussions on the value of $q_{\geq 3} + q_{2}$.
    First, we claim that $q_{\geq 3} + q_{2} \geq 2$ if $\formula$ is reduced.
    The reason is as follows.
    Let $F_1 = S_0\cup S_1\cup S_2$ and $F_2 = \formula\setminus F_1$.
    Recall that by definition, $N(x)\cup\{x\}\subseteq \var(F_1)$, and for $y\in N(x)$, $Q(y):=N(y)\setminus (N(x)\cup \{x\})$.
    Thus, we have $\var(F_1)\cap \var(F_2) = \bigcup_{y\in N(x)}Q(y)$, and so $\card{\var(F_1)\cap \var(F_2)}\leq \sum_{y\in N(x)}\card{Q(y)}=q_2+q_{\geq 3}$.
    Since $\card{\var(F_1)} = 1 + d\leq 8$, R-Rule~\ref{rrule:disjoint} can be applied if $q_{\geq 3} + q_{2} = 0$, and R-Rule~\ref{rrule:cut} can be applied  if $q_{\geq 3} + q_{2} = 1$.
    Next, we consider the following two cases.
    
    \textbf{Case~1: $q_{\geq 3}+q_{2}=2$.}
    Recall that by \eqref{eq:S_2}, we have $(p_{\geq 3} - p_2)\bmod{2}=0$.
    Since $\beta = p_{\geq 3} - p_2 + 2(p_2 + q_2) + 2q_{\geq 3}$, we have $\beta \bmod{2} = 0$.
    On the other hand, we also have $\beta=\sum_{2\leq i\leq d}(i-1)|N_i(x)| + q_{\geq 3} + q_2$ by \eqref{eq:beta}.
    Since $(q_{\geq 3}+q_{2})\bmod{2}=0$ in this case, we can know that $\sum_{2\leq i\leq d}(i-1)|N_i(x)|\bmod{2}=0$ holds.
    Thus, it holds that 
    \[
        \frac{1}{2}\beta=\frac{1}{2}\left(\sum_{2\leq i\leq d}(i-1)|N_i(x)| + 2\right)=\frac{1}{2}\left\lceil \sum_{2\leq i\leq d}(i-1)|N_i(x)|\right\rceil + 1.
    \]

    \textbf{Case~2: $q_{\geq 3}+q_{2}\geq 3$.} 
    In this case, we immediately have 
    \[
        \frac{1}{2}\beta\geq
        \frac{1}{2}\left(\sum_{2\leq i\leq d}(i-1)|N_i(x)| + 3\right)
        \geq \left\lceil\frac{1}{2}\sum_{2\leq i\leq d}(i-1)|N_i(x)|\right\rceil + 1.
    \]
This completes the proof.
\end{proof}

With Lemma~\ref{lem:2cnf-general-bound} in hand, we proceed to determine the branching vectors of our branching operations.

\begin{lemma}\label{lem:alg2step-high-deg}
    In $\algtwocnf$, the branching operation in Line~\ref{line:alg2-branch-1} generates a branching vector not worse than $(5, 11)$.
\end{lemma}
\begin{proof}
    In this branching operation, we branch on a $d$-variable with $d\geq 5$. 
    By Lemma~\ref{lem:2cnf-general-bound}, we have $\Delta_{t},\Delta_{f}\geq d + \card{N_2(x)}\geq d$.
    If $d\geq 8$, this gives a branching vector not worse than $(8, 8)$.
    Otherwise if $5\leq d\leq 7$, by Lemma~\ref{lem:2cnf-general-bound}, we further have 
    % \smalltext{
    \begin{align*}
        \Delta_{t}+\Delta_{f}
        &\geq 2d + \card{N_2(x)} + \left\lceil \frac{1}{2}\sum_{2\leq i\leq d}(i-1)|{N_i(x)}| \right\rceil + 1\\
        &\geq 2d + \frac{1}{2}\left(2\card{N_2(x)}+\sum_{2\leq i\leq d}(i-1)|{N_i(x)}|\right) + 1\\
        &\geq 2d + \frac{1}{2}\left(\sum_{2\leq i\leq d}2|{N_i(x)}|\right) + 1\\
        &\geq 2d + \sum_{2\leq i\leq d}\card{N_{i}(x)}+1\\ 
        &= 3d+1 \geq 16.
    \end{align*}
    % }
    With $\Delta_{t}, \Delta_{f}\geq d\geq 5$, by Lemma~\ref{lem:vec-dominate}, the branching vector not worse than $(5, 11)$.
    Since $\tau(8, 8) < 1.0905 < \tau(5, 11)<1.0956$, the lemma holds.
\end{proof}

\begin{lemma}\label{lem:alg2step-high-neigh}
    In $\algtwocnf$, the branching operation in Line~\ref{line:alg2-branch-2} generates a branching vector not worse than $(4, 11)$. 
\end{lemma}
\begin{proof}
    In this branching operation, we branch on a $4$-variable $x$ with $\card{N_{4}(x)}\geq 3$. 
    Note that in this step we have $\deg(\formula)=4$.
    By Lemma~\ref{lem:2cnf-general-bound} with $d= 4$, we have
    $\Delta_{t}, \Delta_{f}\geq d = 4$ and
    {
    \begin{align*}
        \Delta_{t}+\Delta_{f}
        &\geq 2d + \card{N_2(x)} + \left\lceil \frac{1}{2}\sum_{2\leq i\leq d}(i-1)|{N_i(x)}| \right\rceil + 1\\
        &\geq 2d + \left\lceil\frac{1}{2} \left(2\card{N_2(x)}+\sum_{i=3}^{d}(i-1)|{N_i(x)}|\right)\right\rceil + 1\\
        &\geq 2d + \left\lceil\frac{1}{2} \left(\card{N_{4}(x)}+\sum_{2\leq i\leq d}2|{N_i(x)}|\right)\right\rceil + 1\\
        &\geq 2d + \left\lceil\frac{1}{2} \left(3+2d\right)\right\rceil + 1
        = 15.
    \end{align*}
    }
    By Lemma~\ref{lem:vec-dominate}, the branching vector generated by this step is not worse than $(4, 11)$.
\end{proof}

Next, we analyze the phase three (Lines~\ref{line:alg2-pd}--\ref{line:alg2-call-pw}) of Algorithm~\ref{alg:2CNF}.
\begin{lemma}\label{lem:2cnf-pw}
    Phase three (lines~\ref{line:alg2-pd}-\ref{line:alg2-call-pw}) of $\algtwocnf$ can be excuted in $\bigOs(1.1082^m)$ time.
\end{lemma}
\begin{proof} 
When the algorithm reaches Line~\ref{line:alg2-pd}, the fomrula $\formula$ is a reduced 2-CNF formula with $d(\formula)\leq 4$ such that for every $4$-variable $x$, $|N_4(x)|\leq 2$. 

Let $n:=n(\formula)$, $m:=m(\formula)$, and $n_i$ (resp., $n_{\geq i}$) be the number of variables with degree $i$ (resp., with degree $\geq i$) in $\formula$, where $i\in \mathbb{Z}$. 
Consider the primal graph $G(\formula)$ of formula $\formula$.
Note that $n$ is also the number of vertices in $G(\formula)$, and $n_i$ (resp., $n_{\geq i}$) is also the number of vertices with degree $i$ in $G(\formula)$.
By Lemma~\ref{lem:2cnf-basic-prop}, we have $n_{1}=0$.
Since $d(\formula)\leq 4$, we have $n_{\geq 5}=0$ and
{
\begin{equation*}
    m = \frac{2n_2 + 3n_3 + 4n_4}{2} = \frac{3}{2}(n_3 + 2n_4) - n_4 + n_2.
\end{equation*}
}
By rearranging the above equation, we get
\begin{equation}\label{eq:n3n4}
    n_3 + 2n_4 = \frac{2}{3}(m + n_4 - n_2)\leq \frac{2}{3}(m+n_4).
\end{equation}

Let $V_4$ be the set of $4$-variables in the formula.
The number of clauses that contain two $4$-variables is $\frac{1}{2}\sum_{x\in V_4}\card{N_4(x)}$, and the number of clauses that contain at most one $4$-variable is at least $\sum_{x\in V_4}(4-\card{N_4(x)})$.
Thus, we have
{
\begin{align*}
    m
    \geq \frac{1}{2}\sum_{x\in V_4}\card{N_4(x)} + \sum_{x\in V_4}(4-\card{N_4(x)})
    \geq \sum_{x\in V_4}(4-\frac{1}{2}\card{N_4(x)})
    \geq \sum_{x\in V_4}(4-1) = 3n_4,
\end{align*}
}
which means $n_4\leq m / 3$. 
By putting this into \eqref{eq:n3n4}, we have
\begin{equation}\label{eq:n3+2n4}
    n_3 + 2n_4 \leq \frac{2}{3}(m+n_4)\leq \frac{8}{9}m.
\end{equation}

Let $\epsilon$ be a small constant (say $10^{-9}$) and $n_{\epsilon}$ be the corresponding integer (which is also a constant) in Theorem~\ref{thm:pw-lowd}. 
Let $n:=n(\formula)$ and $m:=m(\formula)$.
If $n \leq n_{\epsilon}$, we invoke algorithm \algbf{} to solve the problem in constant time.
Otherwise, if $n>n_{\epsilon}$, we can apply Theorem~\ref{thm:pw-lowd} and get
\begin{align*}
    \pw(G(\formula)) 
    & \leq n_3/6+n_4/3 + n_{\geq 5} + \epsilon n &\\
    & = (n_3 + 2n_4)/6 + \epsilon n &\\ 
    &\leq \left(4/27+\epsilon\right)m &\text{by \eqref{eq:n3+2n4}}.
\end{align*}
Moreover, by Theorem~\ref{thm:pw-lowd}, a path decomposition of $G(\formula)$ with width at most $\left(4/27+\epsilon\right)m$ can be constructed in polynomial time.
Then, we can apply Theorem~\ref{thm:alg-primal-pw} to solve the problem in time 
\[
    \bigOs(2^{\left(4/27+\epsilon\right)m})\subseteq \bigOs(1.1082^m).
\]
This completes the proof.
\end{proof}

Now we are ready to conclude a running-time bound of Algorithm~$\algtwocnf$.
By Lemma~\ref{lem:alg2step-high-deg} and Lemma~\ref{lem:alg2step-high-neigh}, branching operations in Line~\ref{line:alg2-branch-1} and Line~\ref{line:alg2-branch-2} generate a branching vector not worse than $(5, 11)$ and $(4, 11)$, respectively.
By Lemma~\ref{lem:2cnf-pw}, phase three (Lines~\ref{line:alg2-pd} and \ref{line:alg2-call-pw}) takes $\bigOs(1.1082^m)$ time.
Since $\tau(5, 11) < 1.0956$ and $\tau(4, 11) < 1.1058$, we have the following result.
\begin{theorem}
    $\algtwocnf$ (Algorithm~\ref{alg:2CNF}) solves WMC on 2-CNF formulas in $\bigOs(1.1082^m)$ time, where $m$ is the number of clauses in the input formula.    
\end{theorem}

\section{The Algorithm for Weighted \#3-SAT}
Our algorithm for WMC on $3$-CNF formulas is called $\algthreecnf$ and presented in Algorithm~\ref{alg:3CNF}.
% The sketch of the algorithm is as follows. 
The first phase is also to apply reduction rules to get a reduced instance. 
Note that by Lemma~\ref{lem:reduction-safe}, a reduced formula is still a 3-CNF.
The second phase is to branch on all $3^+$-variables.
When the maximum degree of the formula is at most $2$, we compute a path decomposition of the dual graph of the formula and then invoke the algorithm $\algdualpw$ to solve the problem.

\begin{algorithm}[H]
    \caption{$\algthreecnf\longins$}
    \label{alg:3CNF}
    \textbf{Input}: 3-CNF formula $\formula$, weight function $\wfunc$, and integer $\solW$.\\
    \textbf{Output}: The weighted model count $\solW\cdot \WMC(\formula, \wfunc)$.
    
    \begin{algorithmic}[1] %[1] enables line numbers
        \State Apply reduction rules exhaustively to reduce $\formula$ and update $\wfunc$ and $\solW$ accordingly.
        \IfSingleLine{$\formula$ is empty}{
            \Return $\solW$.
        }
        \IfSingleLine{$\formula$ contains empty clause}{
            \Return $0$.
        }
        \If{there is a $d$-variable $x$ in $\formula$ with $d\geq 3$}
            \State Branch on $x$.\label{line:alg3-branch}
        \Else 
            \State $P\gets$ path decomposition of $G^d(\formula)$ via Theorem~\ref{thm:pw-lowd}. \label{line:alg3-pd}
            \State $\solW_{pw}\gets \algdualpw(\formula, \wfunc, P)$; \label{line:alg3-call-pw}
            \State \Return $\solW\cdot \solW_{pw}$
        \EndIf
    \end{algorithmic}
    
\end{algorithm}

\subsection{The Measure}
The analysis of the algorithm is different from that of the algorithm for \#2-SAT.
In this algorithm, it may not be effective to use $m(\formula)$ as the measure in the analysis since we can not guarantee this measure always decreases in all our steps.
For example, a variable $x$ may only appear as a positive literal in some $3$-clauses.
After assigning $x=0$, it is possible that no reduction rule is applicable and no clause is removed (\ie, $m(\formula) = m(\reduced{\formula[x=0]})$). One of our strategies is to use the following combinatorial measure to analyze the algorithm
\[
    \mu(\formula):=m_3(\formula)+\alpha\cdot m_2(\formula),
\]
where $m_i(\formula) (i\in \{2, 3\})$ is the number of $i$-clauses in formula $\formula$ and $0<\alpha< 1$ is a tunable parameter.
Note that $m(\formula) = m_3(\formula) + m_2(\formula)$ since there is no $1$-clause in a reduced formula by Lemma~\ref{lem:reduced_prop_basic} (we can simply assume that the initial input formula is reduced).
Thus, $\mu(\formula)\leq m(\formula)$.
It can be verified that all the reduction rules would not increase $\mu(\formula)$ for any $0<\alpha<1$.
If we can get a running time bound of $\bigOs(c^{\mu(\formula)})$, with a real number $c>1$, we immediately get a running time bound of $\bigOs(c^{m(\formula)})$.
We will first analyze the algorithm and obtain the branching vectors related to $\alpha$, and then set the value of $\alpha$ to minimize the largest factor.

\subsection{The Analysis}
We first analyze lower bounds for the decrease of the measure $\mu(\formula)$ in a branching operation.
\begin{lemma}\label{lem:3cnf-general-bound}
    Let $\formula$ be a reduced 3-CNF formula, $x$ be a variable in $\formula$, and $c_k (k\in \{2, 3\})$ be the number of $k$-clauses containing variable $x$.
    Let $\Delta_{t}:=\mu(\formula)-\mu(\reduced{\formula[x=1]})$ and $\Delta_{f}:=\mu(\formula)-\mu(\reduced{\formula[x=0]})$.
    It holds that
    \begin{enumerate}
        \item $\Delta_t, \Delta_f\geq c_2\cdot \alpha + c_3 \cdot (1-\alpha)$;
        \item $\Delta_t + \Delta_f\geq c_2\cdot 2\alpha + c_3\cdot (2-\alpha)$.
    \end{enumerate}
\end{lemma}
\begin{proof}
    Let $S_{k}^{\ell}$, where $k\in\{2, 3\}$ and $\ell\in \{x, \nl{x}\}$, be the set of $k$-clauses that contain literal $\literal$.
    By definition, we have $c_{k} = \card{S_{k}^{x}} + \card{S_{k}^{\nl{x}}}$ for $k\in \{2, 3\}$.

    Consider what happens after we assign $x=1$. 
    First, all clauses containing literal $x$ (\ie, clauses in $S_{3}^{x}$ and $S_{2}^{x}$) are satisfied (and so removed). 
    This decreases $\Delta_{t}$ by at least $\card{S_{3}^{x}} + \card{S_{2}^{x}}\cdot \alpha$.
    Second, all $3$-clauses that contain literal $\nl{x}$ (\ie, clauses in $S_{3}^{\nl{x}}$) become $2$-clauses.
    This decreses $\Delta_{t}$ by at least $\card{S_{3}^{\nl{x}}}\cdot(1-\alpha)$.
    Third, all $2$-clauses that contain literal $\nl{x}$ (\ie, clauses in $S_{2}^{\nl{x}}$) become $1$-clauses, and then R-Rule~\ref{rrule:1cls} would be applied to remove these clauses. 
    This decreases $\Delta_{t}$ by at least $\card{S_{2}^{\nl{x}}}\cdot \alpha$.
    In summary, we have
    \begin{align*}
        \Delta_{t} 
        &\geq \card{S_{3}^{x}} + \card{S_{2}^{x}}\cdot \alpha 
                + \card{S_{3}^{\nl{x}}}\cdot(1-\alpha) 
                + \card{S_{2}^{\nl{x}}}\cdot \alpha\\
        &=  \card{S_{3}^{x}}
            + (\card{S_{2}^{x}} + \card{S_{2}^{\nl{x}}})\cdot \alpha
            + (c_3 - \card{S_{3}^{{x}}})\cdot(1-\alpha) \\
        &= c_2\cdot \alpha + c_3 \cdot (1-\alpha) + \card{S_{3}^{{x}}}\cdot\alpha.
    \end{align*}
    Analogously, we have $\Delta_{f}\geq c_2\cdot \alpha + c_3 \cdot (1-\alpha) + \card{S_{3}^{\nl{x}}}\cdot\alpha$.
    Thus, $\Delta_{t}, \Delta_{f}\geq c_2\cdot \alpha + c_3 \cdot (1-\alpha)$ and
    \begin{align*}
        \Delta_{t} + \Delta_{f} 
        & \geq 2\cdot c_2\cdot \alpha + 2\cdot c_3 \cdot (1-\alpha) + \card{S_{3}^{\nl{x}}}\cdot\alpha + \card{S_{3}^{x}}\cdot\alpha\\
        & = c_2\cdot 2\alpha + c_3 \cdot 2(1-\alpha) + c_3\cdot\alpha\\
        & = c_2\cdot 2\alpha + c_3 \cdot (2-\alpha).
    \end{align*}
    This completes the proof.
\end{proof}

Armed with Lemma~\ref{lem:3cnf-general-bound}, we can derive the branching vectors generated by phase two (Line~\ref{line:alg3-branch}) in Algorithm~\ref{alg:3CNF}.
\begin{lemma}\label{lem:3cnf-branch-vectors}
    In $\algthreecnf$, the branching operation in Line~\ref{line:alg3-branch} generates a branching vector not worse than the worst one of the following branching vectors:
    \[
        (3, 3-3\alpha), (2+\alpha, 2-\alpha), (1+2\alpha, 1+\alpha), \mbox{ and } (3\alpha, 3\alpha).
    \]
\end{lemma}
\begin{proof}
    Let $x$ be a $d$-variable with $d\geq 3$.
    Let $p:=c_2\cdot 2\alpha + c_3\cdot (2-\alpha)$ and $q:= c_2\cdot \alpha + c_3 \cdot (1-\alpha)$.
    By Lemma~\ref{lem:3cnf-general-bound}, we have 
    $\Delta_t + \Delta_f\geq p$ and $\Delta_t, \Delta_f\geq q$.
    With Lemma~\ref{lem:vec-dominate}, we know that the branching vector is not worse than $(q,p-q)=(c_2\cdot \alpha + c_3 \cdot (1-\alpha), c_2\cdot \alpha + c_3)$.
    It is evident that larger $c_2$ and $c_3$ result in superior branching vectors.
    Since $c_2+c_3=d\geq 3$, it suffices to consider the case where $c_2+c_3=3$.
    By enumerating all four possible configurations of $c_2$ and $c_3$, we obtain the results stated in the lemma.
\end{proof}

Next, we analyze the time complexity of phase three (Lines~\ref{line:alg3-pd} and \ref{line:alg3-call-pw}) in Algorithm~\ref{alg:3CNF}.
\begin{lemma}\label{lem:3CNF-pw-step}
    Phase three (Lines~\ref{line:alg3-pd}-\ref{line:alg3-call-pw}) of $\algthreecnf$ can be executed in $\bigOs (1.1225^{\frac{\mu(\formula)}{\alpha}})$ time.
\end{lemma}
\begin{proof}
    When the algorithm reaches Line~\ref{line:alg3-pd}, the branching operation is not applicable.
    Thus, at this point, the formula $\formula$ is a reduced 3-CNF formula with $\deg(\formula)\leq 2$. 
    
    Let $m:=m(\formula)$ and $\mu:=\mu(\formula)$. Consider the dual graph $G^d(\formula)$ of formula $\formula$. The number of vertices in $G^d(\formula)$ is $m$.
    
    Let $C\in \formula$ be a clause in formula $\formula$.
    We have $\card{C}\leq 3$.
    Since each variable in $C$ has a degree of at most two, the number of clauses that share a common variable with $C$ is at most $\card{C}\leq 3$.
    That is, for any $C\in \formula$, we have $\card{\{D\in \formula \mid D\neq C \mbox{ and } \var(C)\cap\var(D)\neq \emptyset \}}\leq 3$.

    This means that in $G^d(\formula)$, each vertex has a degree of at most three.
    Let $n_i (i\in \mathbb{Z})$ be the number of vertices with degree $i$ in $G^d(\formula)$.
    We have $n_i=0$ for $i\geq 4$ and $n_3\leq m$.
    Let $\epsilon$ be a small const (say $10^{-9}$) and $m_{\epsilon}$ be the corresponding integer (which is also a constant) in Theorem~\ref{thm:pw-lowd}. 
    If $m \leq m_{\epsilon}$, we invoke brute-force algorithm \algbf{} to solve the problem in constant time.
    Otherwise, if $m>m_{\epsilon}$, we can apply Theorem~\ref{thm:pw-lowd} and get
    \begin{align*}
        \pw(G^d(\formula))
        &\leq n_3/6+n_4/3+ n_{\geq 5} + \epsilon m\\
        &\leq \left(1/6+\epsilon\right) m\leq \left(1/6+\epsilon\right) \frac{\mu}{\alpha}.
    \end{align*}
    Here, the last inequality follows from $m\leq \frac{1}{\alpha}\mu$, which can be derived by the definition of $m$ and $\mu$.
    In addition, by Theorem~\ref{thm:pw-lowd}, a path decomposition of $G^d(\formula)$ with width at most $\left(1/6+\epsilon\right)\frac{\mu}{\alpha}$ can be constructed in polynomial time.
    Then, we can apply Theorem~\ref{thm:alg-dual-pw} to solve the problem in time 
    \[
    \bigOs(2^{\frac{\left(1/6+\epsilon\right)\mu}{\alpha}})\subseteq \bigOs(1.1225^{\frac{\mu}{\alpha}}).
    \]
    This completes the proof.
\end{proof}

We are now poised to analyze the overall running time of $\algthreecnf$. 
The time complexity of phase two (by Lemma~\ref{lem:3cnf-branch-vectors}) and phase three (by Lemma~\ref{lem:3CNF-pw-step}) are summarized in Table~\ref{tb:bvecs}. 

\begin{table}[htbp]
    \centering
    \begin{tabular}{lcc}
    \toprule
        \multirow{2}{*}{Phases} & \multirow{2}{*}{Branching vectors} & {Factors / Base} \\
                               &                                    & $\alpha=0.6309297$ \\
    \midrule
        \multirow{4}{*}{Phase two}   
          & $(3, 3-3\alpha)$ & \textbf{1.4423} \\
          & $(2+\alpha, 2-\alpha)$ & 1.4324 \\
          & $(1+2\alpha, 1+\alpha)$ & 1.4325 \\
          & $(3\alpha, 3\alpha)$ & \textbf{1.4423} \\
    \midrule
        Phase three    & - & $1.1225^{\frac{1}{\alpha}}=1.2011$ \\
    \bottomrule
    \end{tabular}%
    \caption{
        The branching vectors and corresponding factors generated by phase two of $\algthreecnf$, and the base of the time complexity in terms of $\mu(\formula)$ of phase three, all under $\alpha = 0.6309297$.
    }
    \label{tb:bvecs}
\end{table}

By setting $\alpha=0.6309297$, the largest factor in phase two is minimized to $1.4423$, corresponding to the branching vectors $(3, 3-3\alpha)$ and $(3\alpha, 3\alpha)$. 
The time complexity of phase three is $\bigOs(1.1225^\frac{\mu(\formula)}{\alpha})\subseteq\bigOs(1.2011^{\mu(\formula)})$. 
With $\mu(\formula)\leq m(\formula)$, we arrive at the following result.

\begin{theorem}
    $\algthreecnf$ (Algorithm~\ref{alg:3CNF}) solves WMC on 3-CNF formulas in $\bigOs(1.4423^m)$ time, where $m$ is the number of clauses in the input formula.    
\end{theorem}

\section{Conclusion and Discussion}
In this paper, we demonstrate that Weighted Model Counting (WMC) on 2-CNF and 3-CNF formulas can be solved in $\bigOs(1.1082^m)$ and $\bigOs(1.4423^m)$ time, respectively, achieving significant improvements over previous results. The trivial barrier of $\bigOs(2^m)$ for WMC on general CNF formulas cannot be overcome unless SETH fails~\citep{journals/talg/CyganDLMNOPSW16}. It remains an open question whether a running time bound of $\bigOs(c^m)$ with a constant $c<2$ can be achieved for WMC on $k$-CNF formulas for any constant $k$.

Our algorithms first use branch-and-search to effectively eliminate certain problem structures (such as high-degree vertices). Once the remaining problem exhibits some favorable structural properties (such as having a small primal pathwidth or dual pathwidth), dynamic programming and other methods are employed to solve the problem. 
This approach may have potential for application in solving other problems. Furthermore, this method holds significant promise in the design of practical algorithms. 
In practical solving, tree decompositions have been employed in various model counters~\citep{conf/cp/DudekPV20,conf/sat/HecherTW20,conf/cp/FichteHZ19,conf/cp/KorhonenJ21,DBLP:journals/tplp/FichteHTW22,journals/ai/FichteHMTW23}.
Therefore, the practicality of this approach warrants further investigation and exploration.

\section*{Acknowledgements}
The work is supported by the National Natural Science Foundation of China, under the grants 62372095, 62172077, and 62350710215.

\bibliographystyle{plainnat}
\bibliography{sharpSATm}

\begin{thebibliography}{34}
\providecommand{\natexlab}[1]{#1}
\providecommand{\url}[1]{\texttt{#1}}
\expandafter\ifx\csname urlstyle\endcsname\relax
  \providecommand{\doi}[1]{doi: #1}\else
  \providecommand{\doi}{doi: \begingroup \urlstyle{rm}\Url}\fi

\bibitem[Bacchus et~al.(2003)Bacchus, Dalmao, and Pitassi]{conf/focs/BacchusDP03}
Fahiem Bacchus, Shannon Dalmao, and Toniann Pitassi.
\newblock Algorithms and complexity results for {\#}sat and bayesian inference.
\newblock In \emph{44th Symposium on Foundations of Computer Science {(FOCS} 2003), 11-14 October 2003, Cambridge, MA, USA, Proceedings}, pages 340--351. {IEEE} Computer Society, 2003.
\newblock \doi{10.1109/SFCS.2003.1238208}.
\newblock URL \url{https://doi.org/10.1109/SFCS.2003.1238208}.

\bibitem[Beigel and Eppstein(2005)]{journals/jal/BeigelE05}
Richard Beigel and David Eppstein.
\newblock 3-coloring in time o(1.3289\({}^{\mbox{n}}\)).
\newblock \emph{J. Algorithms}, 54\penalty0 (2):\penalty0 168--204, 2005.
\newblock \doi{10.1016/J.JALGOR.2004.06.008}.
\newblock URL \url{https://doi.org/10.1016/j.jalgor.2004.06.008}.

\bibitem[Biere et~al.(2021)Biere, Heule, van Maaren, and Walsh]{series/faia/HandbookSAT}
Armin Biere, Marijn Heule, Hans van Maaren, and Toby Walsh, editors.
\newblock \emph{Handbook of Satisfiability - Second Edition}, volume 336 of \emph{Frontiers in Artificial Intelligence and Applications}.
\newblock {IOS} Press, 2021.
\newblock \doi{10.3233/FAIA336}.
\newblock URL \url{https://doi.org/10.3233/FAIA336}.

\bibitem[Chavira and Darwiche(2008)]{journals/ai/ChaviraD08}
Mark Chavira and Adnan Darwiche.
\newblock On probabilistic inference by weighted model counting.
\newblock \emph{Artif. Intell.}, 172\penalty0 (6-7):\penalty0 772--799, 2008.
\newblock \doi{10.1016/J.ARTINT.2007.11.002}.
\newblock URL \url{https://doi.org/10.1016/j.artint.2007.11.002}.

\bibitem[Chu et~al.(2021)Chu, Xiao, and Zhang]{conf/aaai/ChuXZ21}
Huairui Chu, Mingyu Xiao, and Zhe Zhang.
\newblock An improved upper bound for {SAT}.
\newblock In \emph{Thirty-Fifth {AAAI} Conference on Artificial Intelligence, {AAAI} 2021, Thirty-Third Conference on Innovative Applications of Artificial Intelligence, {IAAI} 2021, The Eleventh Symposium on Educational Advances in Artificial Intelligence, {EAAI} 2021, Virtual Event, February 2-9, 2021}, pages 3707--3714. {AAAI} Press, 2021.
\newblock \doi{10.1609/AAAI.V35I5.16487}.
\newblock URL \url{https://doi.org/10.1609/aaai.v35i5.16487}.

\bibitem[Cook(1971)]{conf/stoc/Cook71}
Stephen~A. Cook.
\newblock The complexity of theorem-proving procedures.
\newblock In Michael~A. Harrison, Ranan~B. Banerji, and Jeffrey~D. Ullman, editors, \emph{Proceedings of the 3rd Annual {ACM} Symposium on Theory of Computing, May 3-5, 1971, Shaker Heights, Ohio, {USA}}, pages 151--158. {ACM}, 1971.
\newblock \doi{10.1145/800157.805047}.

\bibitem[Cygan et~al.(2016)Cygan, Dell, Lokshtanov, Marx, Nederlof, Okamoto, Paturi, Saurabh, and Wahlstr{\"{o}}m]{journals/talg/CyganDLMNOPSW16}
Marek Cygan, Holger Dell, Daniel Lokshtanov, D{\'{a}}niel Marx, Jesper Nederlof, Yoshio Okamoto, Ramamohan Paturi, Saket Saurabh, and Magnus Wahlstr{\"{o}}m.
\newblock On problems as hard as {CNF-SAT}.
\newblock \emph{{ACM} Trans. Algorithms}, 12\penalty0 (3):\penalty0 41:1--41:24, 2016.
\newblock \doi{10.1145/2925416}.
\newblock URL \url{https://doi.org/10.1145/2925416}.

\bibitem[Dahll{\"{o}}f et~al.(2005)Dahll{\"{o}}f, Jonsson, and Wahlstr{\"{o}}m]{journals/tcs/DahllofJW05}
Vilhelm Dahll{\"{o}}f, Peter Jonsson, and Magnus Wahlstr{\"{o}}m.
\newblock Counting models for 2sat and 3sat formulae.
\newblock \emph{Theor. Comput. Sci.}, 332\penalty0 (1-3):\penalty0 265--291, 2005.
\newblock \doi{10.1016/J.TCS.2004.10.037}.
\newblock URL \url{https://doi.org/10.1016/j.tcs.2004.10.037}.

\bibitem[Dubois(1991)]{journals/tcs/Dubois91}
Olivier Dubois.
\newblock Counting the number of solutions for instances of satisfiability.
\newblock \emph{Theor. Comput. Sci.}, 81\penalty0 (1):\penalty0 49--64, 1991.
\newblock \doi{10.1016/0304-3975(91)90315-S}.
\newblock URL \url{https://doi.org/10.1016/0304-3975(91)90315-S}.

\bibitem[Dudek et~al.(2020)Dudek, Phan, and Vardi]{conf/cp/DudekPV20}
Jeffrey~M. Dudek, Vu~H.~N. Phan, and Moshe~Y. Vardi.
\newblock {DPMC:} weighted model counting by dynamic programming on project-join trees.
\newblock In Helmut Simonis, editor, \emph{Principles and Practice of Constraint Programming - 26th International Conference, {CP} 2020, Louvain-la-Neuve, Belgium, September 7-11, 2020, Proceedings}, volume 12333 of \emph{Lecture Notes in Computer Science}, pages 211--230. Springer, 2020.
\newblock \doi{10.1007/978-3-030-58475-7\_13}.
\newblock URL \url{https://doi.org/10.1007/978-3-030-58475-7\_13}.

\bibitem[Due{\~{n}}as{-}Osorio et~al.(2017)Due{\~{n}}as{-}Osorio, Meel, Paredes, and Vardi]{conf/aaai/Duenas-OsorioMP17}
Leonardo Due{\~{n}}as{-}Osorio, Kuldeep~S. Meel, Roger Paredes, and Moshe~Y. Vardi.
\newblock Counting-based reliability estimation for power-transmission grids.
\newblock In Satinder Singh and Shaul Markovitch, editors, \emph{Proceedings of the Thirty-First {AAAI} Conference on Artificial Intelligence, February 4-9, 2017, San Francisco, California, {USA}}, pages 4488--4494. {AAAI} Press, 2017.
\newblock \doi{10.1609/AAAI.V31I1.11178}.
\newblock URL \url{https://doi.org/10.1609/aaai.v31i1.11178}.

\bibitem[Fichte et~al.(2019)Fichte, Hecher, and Zisser]{conf/cp/FichteHZ19}
Johannes~Klaus Fichte, Markus Hecher, and Markus Zisser.
\newblock An improved gpu-based {SAT} model counter.
\newblock In Thomas Schiex and Simon de~Givry, editors, \emph{Principles and Practice of Constraint Programming - 25th International Conference, {CP} 2019, Stamford, CT, USA, September 30 - October 4, 2019, Proceedings}, volume 11802 of \emph{Lecture Notes in Computer Science}, pages 491--509. Springer, 2019.
\newblock \doi{10.1007/978-3-030-30048-7\_29}.
\newblock URL \url{https://doi.org/10.1007/978-3-030-30048-7\_29}.

\bibitem[Fichte et~al.(2022)Fichte, Hecher, Thier, and Woltran]{DBLP:journals/tplp/FichteHTW22}
Johannes~Klaus Fichte, Markus Hecher, Patrick Thier, and Stefan Woltran.
\newblock Exploiting database management systems and treewidth for counting.
\newblock \emph{Theory Pract. Log. Program.}, 22\penalty0 (1):\penalty0 128--157, 2022.
\newblock \doi{10.1017/S147106842100003X}.
\newblock URL \url{https://doi.org/10.1017/S147106842100003X}.

\bibitem[Fichte et~al.(2023{\natexlab{a}})Fichte, Berre, Hecher, and Szeider]{journals/cacm/FichteBHS23}
Johannes~Klaus Fichte, Daniel~Le Berre, Markus Hecher, and Stefan Szeider.
\newblock The silent (r)evolution of {SAT}.
\newblock \emph{Commun. {ACM}}, 66\penalty0 (6):\penalty0 64--72, 2023{\natexlab{a}}.
\newblock \doi{10.1145/3560469}.
\newblock URL \url{https://doi.org/10.1145/3560469}.

\bibitem[Fichte et~al.(2023{\natexlab{b}})Fichte, Hecher, Morak, Thier, and Woltran]{journals/ai/FichteHMTW23}
Johannes~Klaus Fichte, Markus Hecher, Michael Morak, Patrick Thier, and Stefan Woltran.
\newblock Solving projected model counting by utilizing treewidth and its limits.
\newblock \emph{Artif. Intell.}, 314:\penalty0 103810, 2023{\natexlab{b}}.
\newblock \doi{10.1016/J.ARTINT.2022.103810}.
\newblock URL \url{https://doi.org/10.1016/j.artint.2022.103810}.

\bibitem[Fomin and Kratsch(2010)]{series/txtcs/FominK10}
Fedor~V. Fomin and Dieter Kratsch.
\newblock \emph{Exact Exponential Algorithms}.
\newblock Texts in Theoretical Computer Science. An {EATCS} Series. Springer, 2010.
\newblock ISBN 978-3-642-16532-0.
\newblock \doi{10.1007/978-3-642-16533-7}.
\newblock URL \url{https://doi.org/10.1007/978-3-642-16533-7}.

\bibitem[Fomin et~al.(2009)Fomin, Gaspers, Saurabh, and Stepanov]{journals/algorithmica/FominGSS09}
Fedor~V. Fomin, Serge Gaspers, Saket Saurabh, and Alexey~A. Stepanov.
\newblock On two techniques of combining branching and treewidth.
\newblock \emph{Algorithmica}, 54\penalty0 (2):\penalty0 181--207, 2009.
\newblock \doi{10.1007/S00453-007-9133-3}.
\newblock URL \url{https://doi.org/10.1007/s00453-007-9133-3}.

\bibitem[Gaspers and Sorkin(2009)]{conf/soda/GaspersS09}
Serge Gaspers and Gregory~B. Sorkin.
\newblock A universally fastest algorithm for max 2-sat, max 2-csp, and everything in between.
\newblock In Claire Mathieu, editor, \emph{Proceedings of the Twentieth Annual {ACM-SIAM} Symposium on Discrete Algorithms, {SODA} 2009, New York, NY, USA, January 4-6, 2009}, pages 606--615. {SIAM}, 2009.
\newblock \doi{10.1137/1.9781611973068.67}.
\newblock URL \url{https://doi.org/10.1137/1.9781611973068.67}.

\bibitem[Gomes et~al.(2021)Gomes, Sabharwal, and Selman]{series/faia/GomesSS21}
Carla~P. Gomes, Ashish Sabharwal, and Bart Selman.
\newblock Model counting.
\newblock In Armin Biere, Marijn Heule, Hans van Maaren, and Toby Walsh, editors, \emph{Handbook of Satisfiability - Second Edition}, volume 336 of \emph{Frontiers in Artificial Intelligence and Applications}, pages 993--1014. {IOS} Press, 2021.
\newblock \doi{10.3233/FAIA201009}.
\newblock URL \url{https://doi.org/10.3233/FAIA201009}.

\bibitem[Hecher et~al.(2020)Hecher, Thier, and Woltran]{conf/sat/HecherTW20}
Markus Hecher, Patrick Thier, and Stefan Woltran.
\newblock Taming high treewidth with abstraction, nested dynamic programming, and database technology.
\newblock In Luca Pulina and Martina Seidl, editors, \emph{Theory and Applications of Satisfiability Testing - {SAT} 2020 - 23rd International Conference, Alghero, Italy, July 3-10, 2020, Proceedings}, volume 12178 of \emph{Lecture Notes in Computer Science}, pages 343--360. Springer, 2020.
\newblock \doi{10.1007/978-3-030-51825-7\_25}.
\newblock URL \url{https://doi.org/10.1007/978-3-030-51825-7\_25}.

\bibitem[Impagliazzo and Paturi(2001)]{journals/jcss/ImpagliazzoP01}
Russell Impagliazzo and Ramamohan Paturi.
\newblock On the complexity of \emph{k}-{SAT}.
\newblock \emph{J. Comput. Syst. Sci.}, 62\penalty0 (2):\penalty0 367--375, 2001.
\newblock \doi{10.1006/jcss.2000.1727}.

\bibitem[Iwama(1989)]{journals/siamcomp/Iwama89}
Kazuo Iwama.
\newblock {CNF} satisfiability test by counting and polynomial average time.
\newblock \emph{{SIAM} J. Comput.}, 18\penalty0 (2):\penalty0 385--391, 1989.
\newblock \doi{10.1137/0218026}.
\newblock URL \url{https://doi.org/10.1137/0218026}.

\bibitem[Korhonen and J{\"{a}}rvisalo(2021)]{conf/cp/KorhonenJ21}
Tuukka Korhonen and Matti J{\"{a}}rvisalo.
\newblock Integrating tree decompositions into decision heuristics of propositional model counters (short paper).
\newblock In Laurent~D. Michel, editor, \emph{27th International Conference on Principles and Practice of Constraint Programming, {CP} 2021, Montpellier, France (Virtual Conference), October 25-29, 2021}, volume 210 of \emph{LIPIcs}, pages 8:1--8:11. Schloss Dagstuhl - Leibniz-Zentrum f{\"{u}}r Informatik, 2021.
\newblock \doi{10.4230/LIPICS.CP.2021.8}.
\newblock URL \url{https://doi.org/10.4230/LIPIcs.CP.2021.8}.

\bibitem[Kutzkov(2007)]{journals/ipl/Kutzkov07}
Konstantin Kutzkov.
\newblock New upper bound for the {\#}3-sat problem.
\newblock \emph{Inf. Process. Lett.}, 105\penalty0 (1):\penalty0 1--5, 2007.
\newblock \doi{10.1016/J.IPL.2007.06.017}.
\newblock URL \url{https://doi.org/10.1016/j.ipl.2007.06.017}.

\bibitem[Lozinskii(1992)]{journals/ipl/Lozinskii92}
Eliezer~L. Lozinskii.
\newblock Counting propositional models.
\newblock \emph{Inf. Process. Lett.}, 41\penalty0 (6):\penalty0 327--332, 1992.
\newblock \doi{10.1016/0020-0190(92)90160-W}.
\newblock URL \url{https://doi.org/10.1016/0020-0190(92)90160-W}.

\bibitem[Narodytska et~al.(2019)Narodytska, Shrotri, Meel, Ignatiev, and Marques{-}Silva]{conf/sat/NarodytskaSMIM19}
Nina Narodytska, Aditya~A. Shrotri, Kuldeep~S. Meel, Alexey Ignatiev, and Jo{\~{a}}o Marques{-}Silva.
\newblock Assessing heuristic machine learning explanations with model counting.
\newblock In Mikol{\'{a}}s Janota and In{\^{e}}s Lynce, editors, \emph{Theory and Applications of Satisfiability Testing - {SAT} 2019 - 22nd International Conference, {SAT} 2019, Lisbon, Portugal, July 9-12, 2019, Proceedings}, volume 11628 of \emph{Lecture Notes in Computer Science}, pages 267--278. Springer, 2019.
\newblock \doi{10.1007/978-3-030-24258-9\_19}.
\newblock URL \url{https://doi.org/10.1007/978-3-030-24258-9\_19}.

\bibitem[Roth(1996)]{journals/ai/Roth96}
Dan Roth.
\newblock On the hardness of approximate reasoning.
\newblock \emph{Artif. Intell.}, 82\penalty0 (1-2):\penalty0 273--302, 1996.
\newblock \doi{10.1016/0004-3702(94)00092-1}.
\newblock URL \url{https://doi.org/10.1016/0004-3702(94)00092-1}.

\bibitem[Samer and Szeider(2010)]{journals/jda/SamerS10}
Marko Samer and Stefan Szeider.
\newblock Algorithms for propositional model counting.
\newblock \emph{J. Discrete Algorithms}, 8\penalty0 (1):\penalty0 50--64, 2010.
\newblock \doi{10.1016/J.JDA.2009.06.002}.
\newblock URL \url{https://doi.org/10.1016/j.jda.2009.06.002}.

\bibitem[Sang et~al.(2005)Sang, Beame, and Kautz]{conf/aaai/SangBK05}
Tian Sang, Paul Beame, and Henry~A. Kautz.
\newblock Performing bayesian inference by weighted model counting.
\newblock In Manuela~M. Veloso and Subbarao Kambhampati, editors, \emph{Proceedings, The Twentieth National Conference on Artificial Intelligence and the Seventeenth Innovative Applications of Artificial Intelligence Conference, July 9-13, 2005, Pittsburgh, Pennsylvania, {USA}}, pages 475--482. {AAAI} Press / The {MIT} Press, 2005.
\newblock URL \url{http://www.aaai.org/Library/AAAI/2005/aaai05-075.php}.

\bibitem[Valiant(1979)]{journals/tcs/Valiant79}
Leslie~G. Valiant.
\newblock The complexity of computing the permanent.
\newblock \emph{Theor. Comput. Sci.}, 8:\penalty0 189--201, 1979.
\newblock \doi{10.1016/0304-3975(79)90044-6}.
\newblock URL \url{https://doi.org/10.1016/0304-3975(79)90044-6}.

\bibitem[Wahlstr{\"{o}}m(2008)]{conf/iwpec/Wahlstrom08}
Magnus Wahlstr{\"{o}}m.
\newblock A tighter bound for counting max-weight solutions to 2sat instances.
\newblock In Martin Grohe and Rolf Niedermeier, editors, \emph{Parameterized and Exact Computation, Third International Workshop, {IWPEC} 2008, Victoria, Canada, May 14-16, 2008. Proceedings}, volume 5018 of \emph{Lecture Notes in Computer Science}, pages 202--213. Springer, 2008.
\newblock \doi{10.1007/978-3-540-79723-4\_19}.
\newblock URL \url{https://doi.org/10.1007/978-3-540-79723-4\_19}.

\bibitem[Xiao(2022)]{conf/ijcai/Xiao22}
Mingyu Xiao.
\newblock An exact maxsat algorithm: Further observations and further improvements.
\newblock In Luc~De Raedt, editor, \emph{Proceedings of the Thirty-First International Joint Conference on Artificial Intelligence, {IJCAI} 2022, Vienna, Austria, 23-29 July 2022}, pages 1887--1893. ijcai.org, 2022.
\newblock \doi{10.24963/IJCAI.2022/262}.
\newblock URL \url{https://doi.org/10.24963/ijcai.2022/262}.

\bibitem[Zhang(1996)]{journals/tcs/Zhang96a}
Wenhui Zhang.
\newblock Number of models and satisfiability of sets of clauses.
\newblock \emph{Theor. Comput. Sci.}, 155\penalty0 (1):\penalty0 277--288, 1996.
\newblock \doi{10.1016/0304-3975(95)00144-1}.
\newblock URL \url{https://doi.org/10.1016/0304-3975(95)00144-1}.

\bibitem[Zhou et~al.(2010)Zhou, Yin, and Zhou]{conf/aaai/ZhouYZ10}
Junping Zhou, Minghao Yin, and Chunguang Zhou.
\newblock New worst-case upper bound for {\#}2-sat and {\#}3-sat with the number of clauses as the parameter.
\newblock In Maria Fox and David Poole, editors, \emph{Proceedings of the Twenty-Fourth {AAAI} Conference on Artificial Intelligence, {AAAI} 2010, Atlanta, Georgia, USA, July 11-15, 2010}, pages 217--222. {AAAI} Press, 2010.
\newblock \doi{10.1609/AAAI.V24I1.7537}.
\newblock URL \url{https://doi.org/10.1609/aaai.v24i1.7537}.

\end{thebibliography}

\clearpage
\appendix
\section{The Issue in the Analysis of the Previous Algorithm}\label{appendix:note}
We give a note on the issue in the analysis of the algorithm for \#3-SAT suggested in \cite{conf/aaai/ZhouYZ10}.
Their algorithm is presented in Algorithm~\ref{alg:previous} (denoted by \texttt{Alg3CNF-prev}) below. We note that \texttt{Alg3CNF-prev} was proposed only for the unweighted version, and some of the terminology they used differs from ours. 
Here, we restated the algorithm within the framework used in this paper. 

\begin{algorithm}[H]
    \caption{$\texttt{Alg3CNF-prev}\longins$}\label{alg:previous}
    \textbf{Input}: 3-CNF formula $\formula$, weight function $w$, integer $\solW$.\\
    \textbf{Output}: Integer $\solW$, which is the weighted model count of $\formula$.
    
    \begin{algorithmic}[1] %[1] enables line numbers
        % \Statex /* Reduction Rules */
        \State Apply reduction rules exhaustively to reduce $\formula$ and update $\wfunc$ and $\solW$ accordingly.
        % \Statex /* Termination Conditions */
        \IfSingleLine{$\formula$ is empty}{
            \Return $\solW$.
        }
        \IfSingleLine{$\formula$ contains empty clause}{
            \Return $0$.
        }
        \If{there is $3$-clause in $\formula$}
            \State Select a variable $x$ with the maximum degree in all $3$-clauses;\label{line:pre-select-v}
            \State Branch on $x$.
            % \Comment{Branching Step~1}
        \Else % \Comment{$\formula$ is a $2$-CNF}
            \State \Return $\algtwocnf\longins$.
        \EndIf
    \end{algorithmic}
    
\end{algorithm}

Their algorithm employs only simple reduction rules, which are a subset of those used in this paper. 
Our primary focus is on the part of analysis of the branching operation.
The idea of the algorithm is to branch on a variable with maximum degree among all variables in $3$-clauses (if it exists).
When there is no $3$-clause, the formula is a $2$-CNF and so the problem can be solved directly by an algorithm for \#2-SAT (e.g., our $\algtwocnf$).

The analysis in~\cite{conf/aaai/ZhouYZ10} is as follows. 
First, the degree of the variable selected in Line~\ref{line:pre-select-v} is at least $2$.
Otherwise, all variables in $3$-clauses are $1$-variables, and some reduction rules are applicable. This is true. 

The issue arises in the subsequent argumentation:
Second, when $x$ is fixed to a value, every clause containing $x$ (or $\nl{x}$) is either removed or simplified as $2$-clauses. Since the degree of $x$ is at least $2$, at least two clauses are removed when we give a fixed value to $x$. Therefore, we have a branching vector of $(2, 2)$, whose corresponding factor is $1.4142$.

However, the above argument may not always be valid. It is not guaranteed that ``at least two clauses are removed when we assign a value to $x$''. For example, a variable $x$ may only appear as a positive literal in some $3$-clauses. After assigning $x=0$, it is possible that no reduction rule is applicable and no clause is removed (\ie, $m(\formula) = m(\reduced{\formula[x=0]})$).
Thus, we may only be able to get a branching vector of $(2, 0)$.
Since the parameter $m$ does not decrease in the sub-branch $x=0$, such an analysis might not yield a valid running time bound.
This is the reason why we adopt a different measure in our analysis of the algorithm $\algthreecnf$.

\end{document}